\documentclass[12pt]{article}

\usepackage[dvipsnames]{xcolor}
\usepackage{amsmath,amsfonts,amsthm,amssymb,graphicx,arydshln,umoline,appendix,newtxtext,newtxmath,subcaption,physics,color,setspace,enumerate,paralist,comment,mathtools,tikz,bm}
\usepackage{booktabs}
\usepackage{natbib}
\setcitestyle{authoryear,open={(},close={)}}

\usepackage[left=1.15in,right=1.15in,top=1.15in,bottom=1.15in]{geometry}
\usepackage{hyperref}
\hypersetup{colorlinks=true, linkcolor=red, citecolor=blue}

\newtheorem{lemma}{Lemma}
\newtheorem*{lemma*}{Lemma}

\newtheorem{proposition}{Proposition}

\newtheorem{corollary}[proposition]{Corollary}

\newtheorem*{claim*}{Claim}

\theoremstyle{definition}

\newtheorem{example}{Example}

\theoremstyle{remark}

\newcommand{\R}{\mathbb{R}}

\newcommand{\calE}{\mathcal{E}}

\newcommand{\argmax}{\mathop{\rm arg~max}\limits}

\DeclareMathOperator{\supp}{supp}

\title{Information Acquisition with $\alpha$-Divergence Costs\thanks{This work is supported by Grant-in-Aid for Scientific Research Grant Numbers 24K04796 and 25K00619.  I~am grateful to Tommaso Denti for valuable suggestions.}}

\author{Takashi Ui\thanks{Email: \texttt{oui@econ.hit-u.ac.jp}. Postal address: HIAS, 2-1 Naka, Kunitachi, Tokyo 186-8601.}\\\small Kanagawa  University, Yokohama 221-8686, Japan\\\small Hitotsubashi  University, Tokyo 186-8601, Japan}

\date{July 2026}
\begin{document}

\maketitle

\begin{abstract}Building on the $f$-information model of \citet{bdp2025}, this paper introduces a one-parameter family of information acquisition models and characterizes optimal information acquisition. This family extends the mutual information model \citep{matejkaMckay2015} while preserving its analytical tractability. The information cost is derived from the $\alpha$-divergence, which nests the KL divergence ($\alpha=-1$), the reverse KL divergence ($\alpha=1$), and the squared Hellinger distance ($\alpha=0$), and is represented in closed form via the $\alpha$-integration of \citet{amari2007}. The optimal choice probabilities belong to the $q$-exponential family, which appears in nonextensive statistical mechanics \citep{tsallis1988} and in the $q$-logit model of traffic route choice \citep{nakayama2013}. This family reduces to the modified logit in the mutual information case \citep{matejkaMckay2015}. We further show that the relationship between payoffs and the set of actions chosen with positive probability in each state changes qualitatively across ranges of $\alpha$.
\end{abstract}
\newpage

\tableofcontents
\newpage
\section{Introduction}

This paper introduces a one-parameter extension of the mutual-information model of rational inattention. Since \citet{sims2003}, mutual information has served as the canonical specification of information costs in rational inattention. In discrete-choice settings, \citet{matejkaMckay2015} show that this specification yields a modified-logit choice rule, making the model highly tractable. This tractability is tied, however, to a specific functional form that imposes particular restrictions on how choice probabilities respond to payoffs. The one-parameter family introduced in this paper retains the tractability of the mutual-information benchmark while relaxing these restrictions parsimoniously.

We build on the $f$-information model of \citet{bdp2025} (henceforth BDP). BDP introduce a general family of information costs based on $f$-divergences. The $f$-information cost is defined as the value of a minimization problem over reference distributions on signal realizations. BDP exploit this variational structure through convex duality and characterize optimal behavior in terms of the convex conjugate of $f$. Their characterization theorem assumes, in particular, that $f$ is co-finite.

Our model specializes BDP's framework to the $\alpha$-divergence. The $\alpha$-divergence is a one-parameter family of $f$-divergences that nests the KL divergence $(\alpha=-1)$, the reverse KL divergence $(\alpha=1)$, and the squared Hellinger distance $(\alpha=0)$ as special cases, and is central to information geometry \citep{amari2016}. We refer to the resulting $f$-information cost as $\alpha$-information. Since co-finiteness fails for $\alpha>-1$, BDP's characterization theorem does not directly apply in this range. We therefore derive the $\alpha$-information in closed form using the $\alpha$-integration of \citet{amari2007} and use this representation to establish optimality conditions directly.\footnote{The $\alpha$-divergence case is a power-divergence specialization of Csisz\'ar $f$-informativity \citep{csiszar1972}, i.e., the symmetric separable case of BDP's $f$-information. Its minimizing reference marginal is closely related to the power-mean center studied in the literature on information radius \citep{sibson1969,nakiboglu2019}.}

The main result of this paper is a
necessary and sufficient condition for the optimal choice rule, nesting the
optimality condition of \citet{matejkaMckay2015} and
\citet{caplinDeanLeahy2019} for the mutual information cost as a special case.
The optimal choice probabilities belong to the $q$-exponential family, a class
of distributions defined in terms of the $q$-exponential function introduced
by \citet{tsallis1988} in nonextensive statistical mechanics, where it
replaces the exponential in the Boltzmann distribution. The same family appears
in the $q$-logit analysis of traffic route choice \citep{nakayama2013}. 
When $\alpha=-1$, this family reduces to the modified logit form. 
The $q$-exponential function has the same functional form as the inverse of a CRRA utility function, after identifying the CRRA coefficient with $q$.

Our optimality condition also identifies how $\alpha$ changes the set of actions chosen with positive probability in each state, which we call the conditional choice support.  The union of these supports across states is referred to as the consideration set \citep{caplinDeanLeahy2019}.  For $\alpha<-1$, within the co-finite region, each conditional choice support has a cutoff structure: actions in the consideration set are chosen if and only if their payoffs exceed a state-contingent cutoff.  At $\alpha=-1$, the mutual-information case, each conditional choice support coincides with the consideration set.  This also holds for $-1<\alpha<1$, where co-finiteness fails. 
When $\alpha\geq 1$, however, a qualitatively different pattern emerges.  Actions in the common conditional support (i.e., the intersection of conditional choice supports across states) are chosen with positive probability in every state, 
while actions outside the common conditional support are chosen in a state only if they attain the highest payoff among all actions in that state and that payoff is sufficiently high. This implies that an action outside the common support can yield a higher payoff than every action in the common support and nevertheless receive zero probability because another action yields a still higher payoff. Thus, conditional choice no longer follows a simple payoff-cutoff structure: the model permits the persistent use of common actions together with highly selective, state-dependent choice among actions outside the common support.

To further explore the behavioral role of $\alpha$, we study a guess-the-state problem in which the decision maker places a bet on the state and receives a reward if the bet matches the realized state. The optimality condition yields a response function that links reward levels to choice accuracy. Under the mutual-information benchmark ($\alpha=-1$), accuracy is highly responsive to rewards, whereas larger values of $\alpha$ imply weaker responsiveness. The experimental results of \citet{deanNeligh2023} on rational inattention are therefore consistent with $\alpha>-1$, because subjects are less responsive to incentives than the mutual-information benchmark predicts. We use their individual-level observations to assess the approximate magnitude of $\alpha$ consistent with the data. Allowing the cost scale to differ across subjects while keeping $\alpha$ common, we obtain an estimate of $\alpha$ around $3$, well above the mutual-information benchmark.

Section \ref{model} introduces the model. Section \ref{integration} derives the $\alpha$-information cost in closed form. Section \ref{qexponential} establishes that the optimal choice probabilities belong to the $q$-exponential family. Section \ref{OPT} characterizes the optimal rule. 
Section \ref{sec:guess-state} discusses the optimal rule of a guess-the-state problem. 
Section \ref{BDP} compares our optimality condition with that of BDP. Proofs are collected in the appendix.

\section{Information acquisition}\label{model}

Let $\Theta$ be a finite state space. Each state $\theta\in\Theta$ is drawn according to a full-support prior $\pi\in \Delta(\Theta)$.
Let $A$ be a finite action set.  A stochastic choice rule is a family
$
  P=(P_\theta)_{\theta\in\Theta}$, where $P_\theta\in \Delta(A)$. 
The payoff from action $a$ in state $\theta$ is $u(a,\theta)\in\R$.

For a stochastic choice rule $P$, the $\alpha$-information is defined by
\begin{equation}
  I_\alpha(P)
  =
  \min_{m\in\Delta(A)}
  \sum_{\theta\in\Theta}\pi(\theta)D_\alpha[P_\theta:m],
  \label{eq:Ialpha}
\end{equation}
where $D_\alpha[P_\theta:m]$ is the $\alpha$-divergence from $P_\theta$ to $m$. 
The $\alpha$-divergence \citep{Chernoff1952} is defined as follows, where we adopt the   parameterization in \citet{amari2016}. 
If $\alpha\ne\pm1$, 
\begin{equation}
  D_\alpha[p:q]
  \equiv 
  \frac{4}{1-\alpha^2}
  \left(
    1-
    \sum_{a\in A}
    p(a)^{(1-\alpha)/2}q(a)^{(1+\alpha)/2}
  \right)
\notag 
\end{equation}
for $p,q\in\Delta(A)$, 
where we adopt the conventions $0\cdot+\infty =0$, $1/0=+\infty$, and $-\log 0=+\infty$.
If $\alpha=\pm 1$, 
\begin{align*}
  D_{-1}[p:q]&\equiv \lim_{\alpha\to -1}   D_{\alpha}[p:q]
  =
  \sum_{a\in A}p(a)\log\frac{p(a)}{q(a)},\\
  D_1[p:q]&\equiv \lim_{\alpha\to 1}   D_{\alpha}[p:q]
  =
  \sum_{a\in A}q(a)\log\frac{q(a)}{p(a)},
\end{align*}
which are the KL divergence and the reverse KL divergence, respectively.
The $\alpha$-information is a special case of the $f$-information of \citet{bdp2025} (henceforth, BDP), which is discussed in Section \ref{BDP}. 

An information acquisition problem is 
\begin{equation}
  \max_{P\in\Delta(A)^\Theta}\ 
  \sum_{\theta\in\Theta}\pi(\theta)
  \sum_{a\in A}P_\theta(a)u(a,\theta)
  -\kappa I_\alpha(P),
  \label{eq:BDP-alpha}
\end{equation}
where $\kappa>0$ is a unit cost of information.

To provide an optimality condition, we introduce the auxiliary objective
\begin{equation}
  \Phi(P,m)
  =
  \sum_{\theta\in\Theta}\pi(\theta)
  \sum_{a\in A}P_\theta(a)u(a,\theta)
  -
  \kappa
  \sum_{\theta\in\Theta}\pi(\theta)D_\alpha[P_\theta:m].
\notag  
\end{equation}
Note that 
\begin{equation}
  \max_{m\in\Delta(A)}\Phi(P,m)
  =
  \sum_{\theta}\pi(\theta)\sum_a P_\theta(a)u(a,\theta)
  -\kappa I_\alpha(P)\notag 
\end{equation}
by \eqref{eq:Ialpha}. 
Hence the original problem \eqref{eq:BDP-alpha} can be written as
\begin{equation}
  \max_P\max_m \Phi(P,m)
  =\max_m\max_P \Phi(P,m).\notag
\end{equation}
Thus, to solve \eqref{eq:BDP-alpha}, it is useful to consider 
\begin{equation}
  \max_P \Phi(P,m).\label{e-problem}
\end{equation}

\section{\texorpdfstring{The $\alpha$-information cost}{The alpha-information cost}}\label{integration}


For a positive vector $x\in\mathbb{R}_{++}^{\Theta}$, let
\[
M_\alpha[x]
=
h_\alpha^{-1}
\left[
\sum_{\theta\in\Theta}\pi(\theta)h_\alpha[x_\theta]
\right],
\]
where
\[
h_\alpha[t]
=
\begin{cases}
t^{(1-\alpha)/2} & \text{if } \alpha\neq 1, \\
\log(t) & \text{if } \alpha=1.
\end{cases}
\]
The value $M_\alpha[x]$ is referred to as the (weighted) $\alpha$-mean of $x$ (with respect to $\pi$).\footnote{ 
When $\alpha<1$, this definition extends to nonnegative vectors.}
When $\alpha=-1$, it is the arithmetic weighted mean; when $\alpha=1$, it is the geometric weighted mean. In addition, $\lim_{\alpha\to -\infty}M_\alpha(x)=\max_\theta x(\theta)$ and $\lim_{\alpha\to \infty}M_\alpha(x)=\min_\theta x(\theta)$.

The following proposition due to \citet{amari2007} characterizes the unique solution to \eqref{eq:Ialpha} in terms of the $\alpha$-mean. This solution is referred to as the $\alpha$-integration of $P$ (with respect to $\pi$).
For $P\in\Delta(A)^\Theta$, we write
\[
  U(P)=\bigcup_{\theta\in\Theta}\supp P_\theta,
  \qquad
  C(P)=\bigcap_{\theta\in\Theta}\supp P_\theta .
\]

\begin{proposition}\label{Amari proposition}
Assume either $\alpha<1$, or $\alpha\ge1$ and $C(P)\ne\emptyset$.
Define
\[
  Z_\alpha(P)
  =
  \begin{cases}
  \displaystyle
  \sum_{a\in U(P)}M_\alpha[P(a)]
  & \text{if } \alpha<1,\\[1.2em]
  \displaystyle
  \sum_{a\in C(P)}M_\alpha[P(a)]
  & \text{if } \alpha\ge 1,
  \end{cases}
\]
where $M_\alpha[P(a)]$ is the $\alpha$-mean of
$P(a)=(P_\theta(a))_{\theta\in\Theta}$.
Then the problem \eqref{eq:Ialpha} has a unique solution $m_\alpha^P$ given by
\[
  m_\alpha^P(a)
  =
  \begin{cases}
  \displaystyle
  \frac{M_\alpha[P(a)]}{Z_\alpha(P)}
  & \text{if either }\alpha<1 \text{ and } a\in U(P), \text{ or } \alpha\ge1 \text{ and } a\in C(P),\\
  0 
  & \text{otherwise.}
  \end{cases}
\]
If $\alpha\ge1$ and $C(P)=\emptyset$, then 
$ I_\alpha(P)=+\infty$. 
\end{proposition}

Using this proposition, we obtain the $\alpha$-information in the following closed form:
\[
  I_\alpha(P)
  =
  \begin{cases}
\displaystyle\sum_{\theta\in\Theta}\pi(\theta)
  \sum_{a\in U(P)}
  P_\theta(a)
  \log
  \frac{P_\theta(a)}
       {\sum_{\tau\in\Theta}\pi(\tau)P_\tau(a)} &\text{ if  $\alpha=-1$,}\\
\displaystyle  -\log
  \sum_{a\in C(P)}
  \prod_{\theta\in\Theta}
  P_\theta(a)^{\pi(\theta)} & \text{ if $\alpha=1$ and $C(P)\ne\emptyset$,}\\
 \displaystyle
  \frac{4}{1-\alpha^2}
  \left[
    1-
    Z_\alpha(P)^{(1-\alpha)/2}
  \right]  	 & \text{if either $\alpha<1$ and $\alpha\neq -1$,}\\
  & \text{ \quad or $\alpha>1$ and $C(P)\neq\emptyset$,}\\
+  \infty & \text{otherwise.} 
  \end{cases}
\]

The $\alpha$-information cost differs from the posterior-separable power costs used in the rational-inattention literature. 
Given a choice rule $P$, we denote the marginal probability of $a$ by
$
P_\pi(a)
=
\sum_{\theta\in\Theta}\pi(\theta)P_\theta(a)$.
Whenever $P_\pi(a)>0$, the posterior induced by $a$ is
$
\gamma_a^P(\theta)
\equiv 
{\pi(\theta)P_\theta(a)}/{P_\pi(a)}$. 
A (uniformly) posterior-separable cost function has the form
\begin{equation}\label{UPScost}
C(P)
=
\kappa\ 
\qty(\sum_{a\in  \supp P_\pi}
P_\pi(a)H(\gamma_a^P)-H(\pi)),
\end{equation}
where $\kappa>0$ and $H$ is a convex function assigning a cost to each posterior $\gamma_a^P$ \citep{caplinDeanLeahy2022}.\footnote{BDP show that posterior-separable costs are contained in the broader class of $f$-information costs.} 
For $s\neq 0,1$, let \(H_s\) be given by
\[
H_s(\gamma)
\equiv 
\frac{\sum_{\theta\in\Theta}\gamma(\theta)^s-1}{s(s-1)},
\]
which is negative Tsallis entropy divided by \(s\)
\citep{tsallis1988}.\footnote{As \(s\to 1\),
\(H_s(\gamma)\) converges to
\(\sum_{\theta\in\Theta}\gamma(\theta)\log \gamma(\theta)\),
the negative of Shannon entropy.}
The associated power-form information cost is
\begin{equation}
	\label{power cost}
C_s(P)
=
\kappa
\qty(
\sum_{a\in \supp P_\pi}
P_\pi(a)
\frac{
\sum_{\theta\in\Theta}\gamma_a^P(\theta)^s-1
}{s(s-1)}
-
\frac{
\sum_{\theta\in\Theta}\pi(\theta)^s-1
}{s(s-1)}
),
\end{equation}
which is discussed
and used in \citet{caplinDeanLeahy2017,caplinDeanLeahy2022},
\citet{dewanNeligh2020}, and \citet{deanNeligh2023}, among others. 
In Section \ref{UPS}, we discuss information acquisition with this cost and compare its optimal rule with that under the $\alpha$-information cost.

\section{\texorpdfstring{The $q$-exponential family}{The q-exponential family}}\label{qexponential}

An optimal choice rule $P$ is a solution to \eqref{e-problem} when $m=m_\alpha^P$ by Proposition \ref{Amari proposition}.  
In addition, for each $\theta$, $P_\theta\in \Delta(A)$ is a solution to 
\begin{equation}
\max_{p\in \Delta(A)} 
  \sum_{a\in A}p(a)u(a,\theta)
  -
  \kappa D_\alpha[p:m].
  \label{e-problem state}
\end{equation}
To describe an optimality condition for this problem, 
we use the following function:
\begin{equation}
\exp_q(x)=
\begin{cases}
[1+(1-q)x]_+^{1/(1-q)}, & \text{ if } q\neq1\\
\exp(x) & \text{ if } q=1,
\end{cases}\notag
\end{equation}
where
$
[x]_+=\max\{x,0\}$. 
This function is a deformation of the exponential known as the $q$-exponential, introduced in nonextensive statistical mechanics \citep{tsallis1988}.
On the range where $\exp_q$ is positive, its inverse function is denoted by 
\[
\log_q x
=
\begin{cases}
\dfrac{x^{1-q}-1}{1-q} & \text{ if } q\neq 1,\\
\log x & \text{ if } q=1,
\end{cases}
\]
which has the same functional form as the CRRA utility function.

For $m\in\Delta(A)$, let $\bar \lambda_{\theta m}$ denote the unique solution to 
\begin{equation}
\sum_{a\in S_m}
  m(a)
  \exp_{q_\alpha}\left(
  \frac{u(a,\theta)-\lambda}{\kappa}
  \right)
  =
  1,\label{normalization1}
\end{equation}
where 
\[  q_\alpha\equiv\frac{3+\alpha}{2}.
\]
Using $\bar \lambda_{\theta m}$, the following proposition characterizes a solution to \eqref{e-problem state}.

\begin{proposition}\label{q-logit proposition}
For $m\in\Delta(A)$ and $\alpha\in \R$,  let
$
  S_m\equiv \supp m$. 
For each $\theta\in \Theta$, let $P_\theta\in \Delta(A)$ be a solution to \eqref{e-problem state}. 

\begin{enumerate}
\item Suppose that $\alpha\le -1$ {\em ($q_\alpha\leq 1$)}, or $\alpha>-1$  {\em ($q_\alpha> 1$)} and $S_m=A$.
Then, $P_\theta$ is uniquely given by
\begin{equation}
  P_\theta(a)
  =
  m(a)
  \exp_{q_\alpha}\left(
  \frac{u(a,\theta)-\lambda_\theta}{\kappa}
  \right)\quad \text{ for all } a\in A,
\label{q ex family}	
\end{equation}
where $\lambda_\theta=\bar \lambda_{\theta m}$. 
When $\alpha<-1$, $\supp P_\theta\subseteq S_m$.  When $\alpha\geq -1$, $\supp P_\theta=S_m$.  

\item 
Suppose that $\alpha>-1$ and $S_m\subsetneq A$.
For each $\theta\in\Theta$, define
\begin{equation}
  \bar\lambda_{\theta m}^0
\equiv 
  \max_{a\not\in S_m}u(a,\theta)-\frac{2\kappa}{1+\alpha}.\label{zero m lambda}	
\end{equation}
\begin{enumerate}[\em (a)]
\item If $\bar \lambda_{\theta m}\geq \bar\lambda_{\theta m}^0$, $P_\theta$ is uniquely given by \eqref{q ex family}, where $\lambda_\theta=\bar \lambda_{\theta m}$ (we use the convention $1/0=+\infty$ and $0\cdot +\infty =0$).

\item If $\bar \lambda_{\theta m}< \bar\lambda_{\theta m}^0$, then $P_\theta$ is a solution if and only if 
\begin{equation}
  P_\theta(a)
  =
  \begin{cases}
  \displaystyle
  m(a)
  \exp_{q_\alpha}\left(
  \frac{u(a,\theta)-\bar\lambda_{\theta m}^0}{\kappa}
  \right)
  & \text{ if } a\in S_m,\\
  0
  & \displaystyle \text{ if } a\notin S_m \text{ and } u(a,\theta)<\max_{a\not\in S_m} u(a,\theta),
  \end{cases}\notag 
\end{equation}
and the remaining probability mass is assigned among the actions
$a\notin S_m$ attaining $\max_{a\not\in S_m}u(a,\theta)$.
The solution is unique if this set is a singleton, and otherwise it is not unique.
\end{enumerate}

\end{enumerate}
\end{proposition}

A probability distribution of the form \eqref{q ex family} is referred to as the $q$-exponential family \citep{naudts2009,amariOhara2011} and is relevant for the statistical description of small isolated systems in the field of statistical physics \citep{tsallis2023}.  
When $m(a)$ is the uniform distribution over $A$, the choice rule in this proposition is referred to as the $q$-logit choice rule, which is used in the study of route choice problems \citep{nakayama2013}.

\begin{example}
Let 
$
  \Theta=\{1,2,3\}$, $A=\{a,b,c\}$, $\pi(1)=\pi(2)=\pi(3)=1/3$, $\alpha=3$, and $\kappa=1$, where $q_\alpha=3$, $h_3(t)=t^{-1}$, $\exp_3(x)=[1-2x]_+^{-1/2}$, and 
$
  h_3(\exp_3(x))=[1-2x]_+^{1/2}$.
A payoff function is given by Table \ref{tab:payoff}. 
Fix $m=(0,0,1)\in \Delta(A)$, 
where the order of actions is \((a,b,c)\), and $
  S_m=\operatorname{supp}m=\{c\}$. 
We calculate the solutions to the statewise problems \eqref{e-problem state} using Proposition~\ref{q-logit proposition}.

\begin{table}[t]
\centering
\begin{tabular}{c|ccc}
 & \text{ state }1 & \text{ state }2 & \text{ state }3\\
\hline
\text{ action }$a$ & 2 & 1 & -7\\
\text{ action }$b$ & 1 & 2 & -8\\
\text{ action }$c$ & 0 & 0 & 0\\
\end{tabular}
\caption{State-dependent payoffs}
\label{tab:payoff}
\end{table}

	The equation \eqref{normalization1} is
\[
  m(c)\exp_3(u(c,\theta)-\lambda)=  \exp_3(-\lambda)=1,
\]
which gives
$
  \bar\lambda_{1m}=\bar\lambda_{2m}=\bar\lambda_{3m}=0$. 
Since \(2\kappa/(1+\alpha)=1/2\), equation \eqref{zero m lambda} gives
\[
  \bar\lambda^0_{1m}
  =
  \max\{2,1\}-\frac12
  =
  \frac32,
  \quad
  \bar\lambda^0_{2m}
  =
  \max\{1,2\}-\frac12
  =
  \frac32,
  \quad
  \bar\lambda^0_{3m}
  =
  \max\{-7,-8\}-\frac12
  =
  -\frac{15}{2}.
\]
\begin{itemize}
\item For state \(1\), $\lambda_1
  =
  \max\{\bar\lambda_{1m},\bar\lambda^0_{1m}\}
  =\bar\lambda^0_{1m}=
  3/2$. Thus,
\[
\begin{aligned}
  P_1(c)
  =
  m(c)\exp_3(u(c,1)-\lambda_1)  =
  \exp_3\left(-\frac32\right)  =
  [1+3]^{-1/2}  =
  \frac12,
\end{aligned}
\]
and the residual mass $
  1-P_1(c)=1/2$ is assigned to the unique payoff maximizer among the zero-\(m\) actions.  
Since
$
  u(a,1)=2>1=u(b,1)$, 
we have $P_1=(1/2,0,1/2)$.

\item For state \(2\), $  \lambda_2
  =
  \max\{\bar\lambda_{2m},\bar\lambda^0_{2m}\}
  =\bar\lambda^0_{2m}=
  3/2$. Thus, 
\[
\begin{aligned}
  P_2(c)
  =
  m(c)\exp_3(u(c,2)-\lambda_2)  =
  \exp_3\left(-\frac32\right)  =
  [1+3]^{-1/2}  =
  \frac12,
\end{aligned}
\]
and the residual mass $
  1-P_2(c)=1/2$ is assigned to the unique payoff maximizer among the zero-\(m\) actions.  
Since
$
  u(b,2)=2>1=u(a,2)$, 
we have $P_2=(0,1/2,1/2)$.

\item For state \(3\), 
$
  \lambda_3=\max\{ \bar\lambda_{3m},\bar\lambda^0_{3m}\}=\bar\lambda_{3m}=0$. Thus,   
\[
\begin{aligned}
  P_3(c)
  =
  m(c)\exp_3(u(c,3)-\lambda_3) =
  \exp_3(0) =
  1,
\end{aligned}
\]
which implies $
  P_3=(0,0,1)$.
\end{itemize}
In summary, by Proposition~\ref{q-logit proposition}, the unique solutions to the statewise problems \eqref{e-problem state} at \(m=(0,0,1)\) are $P_1=(1/2,0,1/2)$, $P_2=(0,1/2,1/2)$, and $
  P_3=(0,0,1)$. 
  \end{example}

\section{Optimal stochastic choice}\label{OPT}

\subsection{The optimality condition}
Based on Propositions \ref{Amari proposition} and \ref{q-logit proposition}, 
we provide a necessary and sufficient condition for optimality in the next proposition.\footnote{If $\alpha\ge1$ and $C(P)=\emptyset$, then $I_\alpha(P)=+\infty$, so such a stochastic choice rule cannot be optimal. Hence the next proposition concerns stochastic choice rules with finite $\alpha$-information.}

\begin{proposition}\label{main result prop}
A stochastic choice rule $P$ is optimal if and only if the following conditions
are satisfied.

\begin{enumerate}
\item For each $\theta\in\Theta$, $P_\theta$ is a solution to
\eqref{e-problem state} when $m=m_\alpha^P$, as characterized in
Proposition~\ref{q-logit proposition}. Let $\lambda_\theta$ denote the
constant associated with $P_\theta$ in that characterization.

\item If $\alpha\le 1$, then for every $a\in\supp m_\alpha^P$,
\[
\sum_{\theta\in\Theta}\pi(\theta)
h_\alpha\left[
\exp_{q_\alpha}\left(
\frac{u(a,\theta)-\lambda_\theta}{\kappa}
\right)
\right]
=
\max_{b\in A}
\sum_{\theta\in\Theta}\pi(\theta)
h_\alpha\left[
\exp_{q_\alpha}\left(
\frac{u(b,\theta)-\lambda_\theta}{\kappa}
\right)
\right].
\]

If $\alpha>1$, then for every $a\in\supp m_\alpha^P$,
\[
\sum_{\theta\in\Theta}\pi(\theta)
h_\alpha\left[
\exp_{q_\alpha}\left(
\frac{u(a,\theta)-\lambda_\theta}{\kappa}
\right)
\right]
=
\min_{b\in A}
\sum_{\theta\in\Theta}\pi(\theta)
h_\alpha\left[
\exp_{q_\alpha}\left(
\frac{u(b,\theta)-\lambda_\theta}{\kappa}
\right)
\right].
\]
\end{enumerate}
\end{proposition}

When $\alpha=-1$, the above optimality condition reduces to that of 
\citet{caplinDeanLeahy2019}. 
To see this, note that $q_{-1}=1$, $h_{-1}[t]=t$, and  
$
  m_{-1}^P(a)  =
\sum_{\theta\in\Theta}\pi(\theta)P_\theta(a)$,
which is the marginal choice probability. Thus, 
\[
  \exp_{q_\alpha}\left(
  \frac{u(a,\theta)-\lambda_\theta}{\kappa}
  \right)
  =
  \frac{\exp(u(a,\theta)/\kappa)}
  {\sum_{b\in A}m_{-1}^P(b)\exp(u(b,\theta)/\kappa)} .
\]
Then the first condition implies the modified logit formula
\[
  P_\theta(a)
  =
  \frac{
  m_{-1}^P(a)\exp(u(a,\theta)/\kappa)}
  {\sum_{b\in A}m_{-1}^P(b)\exp(u(b,\theta)/\kappa)} .
\]
The second condition implies, for every $a\in\supp m_{-1}^P$,
\[
\sum_{\theta\in\Theta}\pi(\theta)
\frac{\exp(u(a,\theta)/\kappa)}
{\sum_{b\in A}m_{-1}^P(b)\exp(u(b,\theta)/\kappa)}
=
\max_{c\in A}
\sum_{\theta\in\Theta}\pi(\theta)
\frac{\exp(u(c,\theta)/\kappa)}
{\sum_{b\in A}m_{-1}^P(b)\exp(u(b,\theta)/\kappa)} =1
\]
since 
\[
\sum_{c\in A}m_{-1}^P(c)
\sum_{\theta\in\Theta}\pi(\theta)
\frac{\exp(u(c,\theta)/\kappa)}
{\sum_{b\in A}m_{-1}^P(b)\exp(u(b,\theta)/\kappa)}
=1.
\]
This condition is equivalent to, for all $a\in A$,  
\[
\sum_{\theta\in\Theta}\pi(\theta)
\frac{\exp(u(a,\theta)/\kappa)}
{\sum_{b\in A}m_{-1}^P(b)\exp(u(b,\theta)/\kappa)}
\le1,
\]
which is precisely the necessary and sufficient condition of
\citet{caplinDeanLeahy2019}.

\addtocounter{example}{-1}
\begin{example}[continued]
We have shown that, when $m=(0,0,1)$, the unique solutions to the statewise problems \eqref{e-problem state} are
$
  P_1=\left(1/2,0,1/2\right)$, $
  P_2=(0,1/2,1/2)$, and $
  P_3=(0,0,1)$.  
We now show that this rule is an optimal rule by verifying the conditions in Proposition~\ref{main result prop}. 
Since $\alpha=3>1$ and
$
  C(P)=\bigcap_{\theta\in\Theta}\operatorname{supp}P_\theta=\{c\}$, 
the $\alpha$-integration $m_\alpha^P$ is supported only on $c$; that is, 
$
  m_\alpha^P=(0,0,1)$.
Thus the distribution $m=(0,0,1)$ used in the previous section is indeed the $\alpha$-integration of $P$. This implies that the first condition in Proposition~\ref{main result prop} holds.

It remains to check the second condition.  Since $\alpha>1$, action $c$ must minimize
\begin{equation}
  \sum_{\theta\in\Theta}
  \pi(\theta)
  h_3\left[
    \exp_3(u(d,\theta)-\lambda_\theta)
  \right]	\label{c_value}
\end{equation}
over $d\in A$.  
For $d=c$, the value is
\[
\begin{aligned}
  &
  \frac13 h_3\left[\exp_3\left(-\frac32\right)\right]
  +
  \frac13 h_3\left[\exp_3\left(-\frac32\right)\right]
  +
  \frac13 h_3[\exp_3(0)]
  \\
  &=
  \frac13[1+3]^{1/2}
  +
  \frac13[1+3]^{1/2}
  +
  \frac13[1]^{1/2}
  =
  \frac53.
\end{aligned}
\]
For $d=a$, the value is
\[
\begin{aligned}
  &
  \frac13 h_3\left[\exp_3\left(\frac12\right)\right]
  +
  \frac13 h_3\left[\exp_3\left(-\frac12\right)\right]
  +
  \frac13 h_3[\exp_3(-7)]
  \\
  &=
  \frac13[1-1]_+^{1/2}
  +
  \frac13[1+1]^{1/2}
  +
  \frac13[1+14]^{1/2}
  =
  \frac{\sqrt{2}+\sqrt{15}}{3}.
\end{aligned}
\]
For $d=b$, the value is
\[
\begin{aligned}
  &
  \frac13 h_3\left[\exp_3\left(-\frac12\right)\right]
  +
  \frac13 h_3\left[\exp_3\left(\frac12\right)\right]
  +
  \frac13 h_3[\exp_3(-8)]
  \\
  &=
  \frac13[1+1]^{1/2}
  +
  \frac13[1-1]_+^{1/2}
  +
  \frac13[1+16]^{1/2}
  =
  \frac{\sqrt{2}+\sqrt{17}}{3}.
\end{aligned}
\]
Since
\[
  \frac53
  <
  \frac{\sqrt{2}+\sqrt{15}}{3}
  <
  \frac{\sqrt{2}+\sqrt{17}}{3},
\]
action $c$ uniquely minimizes \eqref{c_value}.
Therefore the stochastic choice rule given by 
$
  P_1=\left(1/2,0,1/2\right)$, $
  P_2=(0,1/2,1/2)$, and $
  P_3=(0,0,1)$ is optimal, where the support of the $\alpha$-integration is $\{c\}$.  
  
 In state $1$,
$
  u(a,1)=2>1=u(b,1)>0=u(c,1)$, 
while
\[
  P_1(a)=\frac12>0,
  \qquad
  P_1(b)=0,
  \qquad
  P_1(c)=\frac12>0.
\]
Thus action $b$ is strictly better than the action $c$ in the support of the $\alpha$-integration in state $1$, but it receives zero probability, while $c$ receives positive probability.  This is because $b$ is not the best zero-$m$ action in that state.
Likewise, in state $2$,
$
  u(b,2)=2>1=u(a,2)>0=u(c,2)$, 
while
\[
  P_2(a)=0,
  \qquad
  P_2(b)=\frac12>0,
  \qquad
  P_2(c)=\frac12>0.
\]
Thus action $a$ is strictly better than the action $c$ in the support of the $\alpha$-integration in state $2$, but it receives zero probability, while $c$ receives positive probability.  This is because $a$ is not the best zero-$m$ action in that state.
\end{example}

\subsection{Conditional choice supports}

We discuss how an optimal choice rule $P$ depends on $\alpha$ by focusing on the support of $P_\theta$ for each $\theta$. 
The support of  $P_\theta$ and its union over states are referred to as the conditional choice support at $\theta$ and the consideration set, respectively. 
The following result characterizes the relationship between conditional choice supports and the $\alpha$-integration.

\begin{corollary}\label{supp coro}
Let $P$ be an optimal stochastic choice rule with the $\alpha$-integration $m=m_\alpha^P$. 
Define \[\nu_\theta \equiv \lambda_\theta +2\kappa/(1+\alpha)\]
for $\alpha\neq -1$, 
where $\lambda_\theta$ is the constant of the optimal choice rule given in Proposition \ref{main result prop}. 
The following statements hold, where $S_m\equiv\supp m$ and $S_\theta\equiv \supp P_\theta$.
\begin{enumerate}
\item If $\alpha<-1$, then
$
  S_\theta\subseteq S_m$ for every $\theta\in\Theta$ and $
  S_m=\bigcup_{\theta\in\Theta}S_\theta$. 
In addition, $a\in S_\theta$ if and only if $u(a,\theta)>\nu_\theta$ and $a\in S_m$.

\item If $-1\leq\alpha<1$, then
$
  S_\theta=S_m$ for every $\theta\in\Theta$. 

\item If $\alpha\geq1$, then $ S_m=\bigcap_{\theta\in\Theta}S_\theta$. 
For every $a\in S_m$ and $\theta\in\Theta$, $u(a,\theta)<\nu_\theta$. 
If $a\in S_\theta\setminus S_m$, then $u(a,\theta)=\nu_\theta$. 
\end{enumerate}
\end{corollary}

By this corollary, conditional choice supports have distinct properties depending on the three ranges of $\alpha$.

\begin{itemize}
\item Suppose that $\alpha<-1$. The consideration set coincides with the support of the $\alpha$-integration. The conditional choice support at $\theta$ consists of actions in the consideration set whose payoffs exceed the cutoff $\nu_\theta$.
\item Suppose that $-1\leq \alpha<1$. Not only the consideration set but also each conditional choice support coincides with the support of the $\alpha$-integration. 

\item Suppose that $\alpha\geq 1$. The intersection of all conditional choice supports coincides with the support of the $\alpha$-integration. The conditional choice support at $\theta$ consists of actions in the $\alpha$-integration and outside actions that attain the highest payoff among all actions, provided that this highest payoff is sufficiently high (cf. Proposition \ref{q-logit proposition}).
\end{itemize}

We illustrate these three regimes using a numerical example. 
Keeping the prior, payoffs, and information-cost scale in Example~1 fixed,
we numerically calculate optimal rules for selected values of 
$\alpha$. The conditional choice supports and the corresponding optimal rules are reported in Tables \ref{tab:support} and \ref{tab:rules}, respectively. 

\begin{table}[t]
\centering
\begin{tabular}{c|c|c|c|c|c}
  $\alpha$ & $q_\alpha$ & $S_m$ & $S_1$ & $S_2$ & $S_3$ \\ \hline
  $-3$ & $0$   & $\{a,b,c\}$ & $\{a,b,c\}$ & $\{a,b,c\}$ & $\{c\}$     \\
  $-1$ & $1$   & $\{a,b,c\}$ & $\{a,b,c\}$ & $\{a,b,c\}$ & $\{a,b,c\}$ \\
  $0$  & $1.5$ & $\{a,b,c\}$ & $\{a,b,c\}$ & $\{a,b,c\}$ & $\{a,b,c\}$ \\
  $1$  & $2$   & $\{a,b,c\}$ & $\{a,b,c\}$ & $\{a,b,c\}$ & $\{a,b,c\}$ \\
  $3$  & $3$   & $\{c\}$     & $\{a,c\}$   & $\{b,c\}$   & $\{c\}$     \\
\end{tabular}
\caption{Support patterns for selected values of $\alpha$.}
\label{tab:support}
\end{table}

\begin{table}[t]
\centering
\begin{tabular}{c|c|c|c}
  $\alpha$ & $P_1$ & $P_2$ & $P_3$ \\ \hline
  $-3$ & $(0.620,\ 0.328,\ 0.052)$ & $(0.328,\ 0.620,\ 0.052)$ & $(0,\ 0,\ 1)$                 \\
  $-1$ & $(0.641,\ 0.235,\ 0.123)$ & $(0.236,\ 0.641,\ 0.123)$ & $(0.001,\ 0.000,\ 0.999)$     \\
  $0$  & $(0.639,\ 0.187,\ 0.174)$ & $(0.211,\ 0.612,\ 0.177)$ & $(0.016,\ 0.011,\ 0.973)$     \\
  $1$  & $(0.613,\ 0.125,\ 0.262)$ & $(0.166,\ 0.566,\ 0.268)$ & $(0.028,\ 0.020,\ 0.952)$     \\
  $3$  & $(0.5,\ 0,\ 0.5)$         & $(0,\ 0.5,\ 0.5)$         & $(0,\ 0,\ 1)$                 \\
\end{tabular}
\caption{Optimal rules for selected values of $\alpha$. Numerically computed entries are rounded to three decimal places; exact entries are displayed without trailing zeros. An entry displayed as $0$ is exactly zero, whereas an entry displayed as $0.000$ is positive before rounding.}
\label{tab:rules}
\end{table}

When $\alpha=-3<-1$, the support of the $\alpha$-integration is $S_m=\{a,b,c\}$.
Since
$
  S_1=S_2=\{a,b,c\}$ and $S_3=\{c\}$, 
the cutoff characterization implies
\[
\begin{array}{lll}
  u(a,1)>\nu_1, & u(b,1)>\nu_1, & u(c,1)>\nu_1, \\[2mm]
  u(a,2)>\nu_2, & u(b,2)>\nu_2, & u(c,2)>\nu_2, \\[2mm]
  u(a,3)\leq\nu_3, & u(b,3)\leq\nu_3, & u(c,3)>\nu_3.
\end{array}
\]
Thus, $a$ and $b$ are not selected in state~$3$ because their payoffs are below the relevant lower
cutoff, while all actions are selected in states~$1$ and $2$ because their payoffs exceed the cutoff in those states.

When $\alpha=-1$ and  $\alpha=0$, 
the support of the $\alpha$-integration is $S_m=\{a,b,c\}$, which coincides with each conditional choice support.

When $\alpha=1$, the support of the $\alpha$-integration is $S_m=\{a,b,c\}$. 
In this case, all actions are selected with positive probability in every state because the support of the $\alpha$-integration coincides with the intersection of all conditional choice supports.

When $\alpha=3$, as established in Example~1, the support of the $\alpha$-integration is $S_m=\{c\}$.
Thus, $c$ is selected in all states. 
In state~$1$, however, an action outside $S_m$ is additionally selected and $S_1=\{a,c\}$.
Since
\[
  u(a,1)=2>1=u(b,1)>0=u(c,1),
\]
action $a$ is the payoff-maximizing action in state~$1$.
Its payoff reaches the upper boundary required for selection outside $S_m$, and it
therefore receives the residual probability mass. Action $b$ is not selected, even though
it yields a higher payoff than $c$, because it is not a payoff maximizer.
In state~$2$, an action outside $S_m$ is additionally selected and $S_2=\{b,c\}$.
Since
\[
  u(b,2)=2>1=u(a,2)>0=u(c,2),
\]
action $b$ is the payoff-maximizing action in state~$2$.
Its payoff reaches the upper boundary required for selection outside $S_m$, and it
therefore receives the residual probability mass. Action $a$ is not selected, even though
it yields a higher payoff than $c$, because it is not a payoff maximizer.
In state~$3$, the activation condition for an action outside
$S_m$ is not satisfied, so $S_3=\{c\}$.

\section{Guess-the-state problem}\label{sec:guess-state}

This section applies Proposition \ref{main result prop} to symmetric guess-the-state problems, which serve as experimentally testable models of rational inattention, and conducts a comparative static analysis with respect to $\alpha$. 
We then assess the approximate magnitude of $\alpha$ consistent with the experimental results of \cite{deanNeligh2023}.

\subsection{Response functions}

Let
$
\Theta=A=\{1,\ldots,n\}$ with 
$n\geq 2$ and $\pi(\theta)=1/n\text{ for all }\theta\in\Theta$. 
The payoff to action $a$ at state $\theta$ is 
\[
u(a,\theta)=
\begin{cases}
w & \text{if } a=\theta,\\
0 & \text{if } a\neq\theta,
\end{cases}
\]
where $w>0$ is constant. 
Thus, the decision maker receives $w$ from guessing the state correctly.

Since the problem is symmetric across states and actions, we restrict attention to a symmetric stochastic choice rule:
\begin{equation}
P_\theta(a)=
\begin{cases}
\rho & \text{if } a=\theta,\\
\dfrac{1-\rho}{n-1} & \text{if } a\neq\theta,
\end{cases}
\notag\label{eq:guess-state-rule}
\end{equation}
where $\rho$ is the probability of choosing a correct action. 
The following proposition characterizes a symmetric optimal choice rule.\footnote{See Proposition 10 of BDP for the co-finite $f$-information case.}

\begin{proposition}\label{prop:guess-state}
There exists a symmetric optimal choice rule. 
The probability of a correct action, denoted by $\rho(w)$ as a function of $w$, is characterized as follows.

\begin{enumerate}
\item Suppose that $\alpha<-1$. Define
\[
\bar w_\alpha
\equiv
\frac{\kappa n^{1-q_\alpha}}{1-q_\alpha}.
\]
If $w<\bar w_\alpha$, then $\rho(w)\in(1/n,1)$ is the unique solution to
\begin{equation}
\left(n\rho(w)\right)^{1-q_\alpha}
-
\left(
\frac{n(1-\rho(w))}{n-1}
\right)^{1-q_\alpha}
=
(1-q_\alpha)\frac{w}{\kappa}.
\label{eq:guess-state-response}
\end{equation}
If $w\geq\bar w_\alpha$, then $\rho(w)=1$.

\item Suppose that $\alpha=-1$. Then,
\begin{equation}
\rho(w)
=
\frac{\exp(w/\kappa)}
{\exp(w/\kappa)+n-1}.\label{eq:guess-state-response-q1}
\end{equation}

\item Suppose that $\alpha> -1$. Then $\rho(w)\in(1/n,1)$ is the unique solution to \eqref{eq:guess-state-response}. 
\end{enumerate}
\end{proposition}

The function $\rho:\mathbb{R}_+\to (0,1]$ in Proposition \ref{prop:guess-state} is referred to as the response function. When $\alpha < -1$, $\rho(w)$ equals one for sufficiently large rewards. When $\alpha \geq -1$, $\rho(w)$ approaches one as the reward tends to infinity. The following corollary characterizes the rate of convergence.\footnote{For functions $f,g$, we write $f(x) = O(g(x))$ as $x \to \infty$ if there exist $C > 0$ and $x_0$ such that $|f(x)| \le C|g(x)|$ for all $x \ge x_0$.}

\begin{corollary}\label{prop:guess-state-asymptotic}
Let $\alpha\geq -1$ and $n\geq 2$. The error probability
\(1-\rho(w)\) satisfies
\[
1-\rho(w)
=
\begin{cases}
O\left(\exp\left(-\dfrac{w}{\kappa}\right)\right)
& \text{ if }\alpha=-1,\\[0.8em]

O\left(\left(\dfrac{w}{\kappa}\right)^{-\frac{2}{1+\alpha}}\right)
& \text{ if }\alpha>-1.
\end{cases}
\]
\end{corollary}

When \(\alpha=-1\), the error probability converges to zero at an exponential rate as $w$ approaches infinity.  When \(\alpha>-1\), the error probability decreases only at a polynomial rate.  Moreover, the polynomial exponent
${2}/({1+\alpha})$ 
is decreasing in \(\alpha\); hence larger values of \(\alpha\) imply slower convergence of the error probability to zero, as illustrated in  
Figure~\ref{fig:alpha-response-functions} which plots $\rho(w)$ when $n=2$ for $\alpha\in\{3,1,-1,-2,-3\}$. 

\begin{figure}
\centering
\resizebox{0.7\textwidth}{!}{%
  \includegraphics{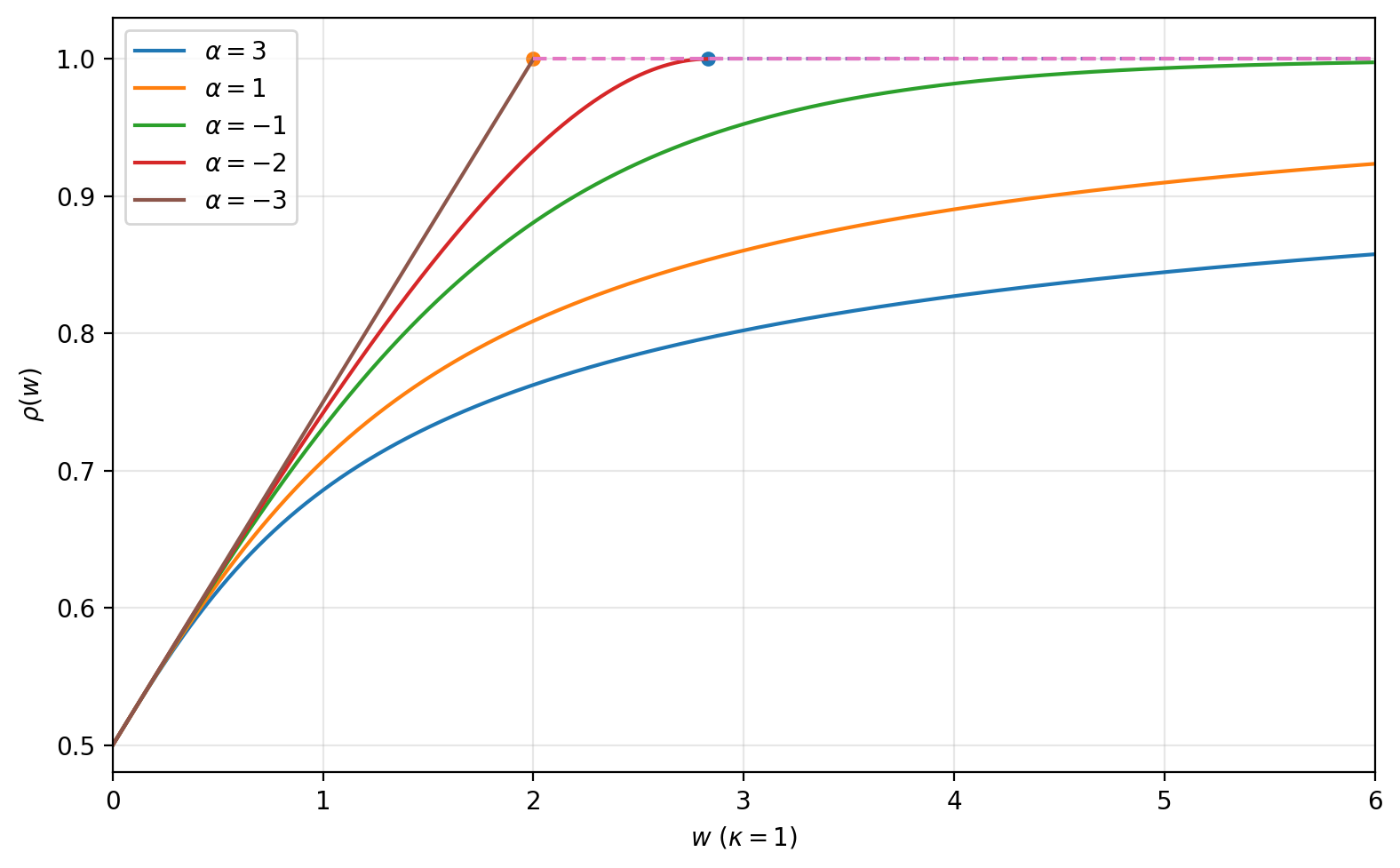}%
}
\caption{Response functions in the symmetric two-state two-action guess-the-state problem. The figure plots $\rho(w)$ against the reward $w$ for $\alpha\in\{3,1,-1,-2,-3\}$ with $\kappa=1$. The dashed horizontal segments indicate the ceiling $\rho(w)=1$  for $\alpha<-1$.}
\label{fig:alpha-response-functions}
\end{figure}

\subsection{Comparative statics}

To further study the behavioral implications of $\alpha$, we focus on the binary case $n=2$ with $\alpha>-1$, and consider a symmetric choice rule $P$ with the probability of a correct choice given by $\rho>1/2$: 
\[
P_1=(\rho,1-\rho),
\qquad
P_2=(1-\rho,\rho).
\]
We write the $\alpha$-information cost of $P$ as a function of $\rho$ and denote it by $c_\alpha(\rho)$. Since the $\alpha$-integration of a symmetric choice rule is the uniform distribution over $A$, we have
\[
c_\alpha(\rho)\equiv 
\frac{4\kappa}{1-\alpha^2}
\left(
1
-
2^{-(1+\alpha)/2}
\left(
\rho^{(1-\alpha)/2}
+
(1-\rho)^{(1-\alpha)/2}
\right)
\right).
\]
Define the elasticity of marginal cost with respect to the probability of a correct choice by
\[
\calE_{\alpha}(\rho)
\equiv 
\frac{d\log c_\alpha'(\rho)}{d\log \rho}
=
\frac{\rho c_\alpha''(\rho)}{c_\alpha'(\rho)}
=
\frac{1+\alpha}{2}
\frac{
\left(\rho/(1-\rho)\right)^{(3+\alpha)/2}
+
1
}{
\left(\rho/(1-\rho)\right)^{(1+\alpha)/2}
-
1
}>0.
\]
The next lemma shows that this elasticity strictly increases with $\alpha$.\footnote{$\calE_{\alpha}(\rho)$ is not necessarily monotone in $\alpha$ when $\alpha<-1$.}

\begin{lemma}\label{lem:elasticity-marginal-cost-alpha}
For each $\rho>1/2$, $\calE_{\alpha}(\rho)$ is strictly increasing in $\alpha$ on $(-1,\infty)$.
\end{lemma}

This lemma implies that a larger $\alpha$ makes it harder to raise the probability of a correct choice by increasing the reward $w$.  
To see this formally, consider the inverse response function $w_\alpha(\rho)$, defined as the reward required to induce the correct-choice probability $\rho$.
Since 
$
w_\alpha(\rho)=c_\alpha'(\rho)$ 
by the first-order condition for an interior optimum, the elasticity of the required reward with respect to the probability of a correct choice equals the elasticity of marginal cost:
\[
\frac{d\log w_\alpha(\rho)}{d\log \rho}
=
\frac{\rho w_\alpha'(\rho)}{w_\alpha(\rho)}
=
\frac{\rho c_\alpha''(\rho)}{c_\alpha'(\rho)}
=
\mathcal E_\alpha(\rho).
\]
Lemma~\ref{lem:elasticity-marginal-cost-alpha} shows that this elasticity is strictly increasing in $\alpha$ for every $\rho>1/2$.
Therefore, at any fixed achieved accuracy $\rho$, a larger $\alpha$ implies that 
a given proportional increase in the probability of a correct choice requires a larger proportional increase in the reward. 
This results in a slow convergence rate of the response function to one as the reward approaches infinity for a larger $\alpha$.

\subsection{Estimated magnitude of \texorpdfstring{$\alpha$}{alpha}}\label{Estimated magnitude sec}

Lemma~\ref{lem:elasticity-marginal-cost-alpha} and the experimental results of rational inattention by \citet{deanNeligh2023} suggest that $\alpha$ is greater than $-1$.
Experiment~1.2 in \citet{deanNeligh2023} considers the same guess-the-state problem as in the previous subsection, in which subjects are less responsive to incentives than the mutual information benchmark predicts.
To conduct statistical analysis, \citet{deanNeligh2023} adopt several alternatives to the mutual-information model, and one specification uses a posterior-separable cost of power form~\eqref{power cost}, which is a power function of posteriors.
This differs from the $\alpha$-information cost, which is a function of choice probabilities.
However, when restricted to symmetric choice rules in the symmetric binary environment, the two costs take the same functional form through reparameterization and normalization, since the choice probability $P_\theta(a)$ equals the posterior $\gamma_a(\theta) = P_\theta(a)\pi(\theta)/\sum_{\theta'} P_{\theta'}(a)\pi(\theta')$ for each $a, \theta \in \Theta = \{1, 2\}$.

\citet{deanNeligh2023} write the exponent on a posterior probability $\gamma_a(\theta)$ as $2-r$.
Their estimate of $r$ from the aggregate choice frequencies in Experiment~1.2 is
\[
\widehat r^{\mathrm{DN}}=13.41.
\]
In the $\alpha$-information cost, the exponent on a choice probability $P_\theta(a)$ is $(1-\alpha)/2$.
Under the uniform prior and a symmetric choice rule, the revealed posterior satisfies $\gamma_a(\theta)=P_\theta(a)$.
Equating the two exponents gives
$
2-r
=
({1-\alpha})/{2}$, 
and hence
\[
\widehat \alpha^{\mathrm{DN}}
=
2\widehat r^{\mathrm{DN}}-3
=
23.82.
\]
This estimate is based on the assumption that all subjects share a common cost function.\footnote{Although \citet{deanNeligh2023} also conduct statistical analyses allowing heterogeneous cost functions, their reported estimate of $r$ is based on the common-cost specification.}

From a different perspective, 
we assess the approximate magnitude of $\alpha$ using the individual-level observations from Experiment~1.2.
A reward for a correct choice is $w\in\{5,40,70,95\}$, measured in probability points.
Let $\rho(w|\alpha,\kappa)$ denote the predicted correct-choice probability at reward $w$, which is the unique solution to \eqref{eq:guess-state-response}. We allow the cost scale $\kappa$ to differ across subjects while keeping $\alpha$ common.
For subject $i$, let $n_{iw}$ and $y_{iw}$ be the number of trials and that of correct choices, respectively, at reward $w$.
For each $\alpha$, we estimate subject $i$'s cost scale $\kappa_i$ by maximizing
\[
\ell_i(\alpha,\kappa_i)
\equiv
\sum_{w\in\{5,40,70,95\}}
\left[
y_{iw}\log \rho(w|\alpha,\kappa_i)
+
(n_{iw}-y_{iw})\log\{1-\rho(w|\alpha,\kappa_i)\}
\right].
\]
Then we estimate $\alpha$ by maximizing the profile log-likelihood 
\[
\ell_p(\alpha)
\equiv
\sum_i
\max_{\kappa_i\geq 0}
\ell_i(\alpha,\kappa_i).
\]
The profile maximum likelihood estimate is\footnote{See Appendix \ref{app:individual-profile-likelihood}.}
\[
\widehat{\alpha}=3.18.
\]

For comparison, we also estimate $\alpha$ under the common-cost specification, in which both $\alpha$ and $\kappa$ are common across subjects and trials.
The maximum likelihood estimate is
\[
\widehat{\alpha}^{\mathrm{cc}}
\simeq
23.95,
\qquad
\widehat{\kappa}^{\mathrm{cc}}
\simeq
0.382.
\]
This value is close to the translated estimate $\widehat\alpha^{\mathrm{DN}}=23.82$, as expected. The estimate under the common-cost specification is much larger because 
flatness in the reward-accuracy relation can arise for reasons other than 
a large value of $\alpha$. For a fixed value of $\alpha$, a subject with a small $\kappa_i$ has a high value of $\rho_i(5)$ (the accuracy at the lowest reward level), so an increase in the reward has little room to raise accuracy. In contrast, a subject with a large $\kappa_i$ has a low value of $\rho_i(95)$ (the accuracy at the highest reward level), so a decrease in the reward has little room to reduce accuracy.
 Thus, with heterogeneous 
$\kappa_i$, flat reward-accuracy relations can arise either because 
accuracy is already high at low rewards or because accuracy remains low 
even at high rewards. Under the common-cost specification, however, a single 
$\kappa$ cannot capture these different sources of flatness. The remaining flatness 
is therefore attributed to $\alpha$, pushing 
$\widehat{\alpha}^{\mathrm{cc}}$ upward.\footnote{This is not a bias in the usual sense. The common-cost estimator targets the curvature of the aggregate response curve, which need not coincide with the common individual-level value of $\alpha$ when cost scales differ across subjects.} 
 Once subject-level heterogeneity 
in $\kappa_i$ is allowed, these sources of flatness are absorbed by the 
subject-specific cost scales, and the estimate of $\alpha$ falls to 
around $3$.





\section{Discussion}\label{BDP}
\subsection{Comparison with BDP}\label{BDPsub}

We relate Proposition~\ref{main result prop} to the optimality condition of BDP.  For a choice rule \(P=(P_\theta)_{\theta\in\Theta}\)
and a reference distribution \(m\in\Delta(A)\), the \(f\)-divergence is
\[
  D_f(P:m)
  =
  \sum_{a\in A}m(a)
  f\left(
    \left(
      \frac{P_\theta(a)}{m(a)}
    \right)_{\theta\in\Theta}
  \right),
\]
where $f:\R_+^\Theta\to \R$,
and the
\(f\)-information is
\[
  I_f(P)
  =
  \inf_{m\in\Delta(A)}D_f(P:m).
\]
An \(f\)-mean of \(P\) is defined as $m\in \Delta(A)$ attaining this infimum.
\citet{csiszar1972} consider the separable case, 
\[
  f(x)
  =
  \sum_{\theta\in\Theta}\pi(\theta)\varphi(x_\theta),
\]
where \(\varphi\) is a univariate convex function.  
BDP refer to the $f$-information in the separable case as the Csisz\'ar information.
BDP impose co-finiteness as part of their regularity assumption on \(f\). 
For the separable case, co-finiteness requires that
\(\lim_{t\to\infty}\varphi(t)/t=+\infty\).

BDP study the information acquisition problem \eqref{eq:BDP-alpha} replacing $\kappa I_\alpha(P)$ with $I_f(P)$. 
In the separable case, the optimality condition is given in terms of 
the convex conjugate of \(\varphi\) denoted by \(\psi\equiv \varphi^\ast\).  
BDP show that $P$ is optimal if and only if there exist \(m\in\Delta(A)\) and $\zeta\in\R^\Theta$ such that 
\begin{gather}
  P_\theta(a)
  =
  m(a)\psi'\left(
    u(a,\theta)-\zeta_\theta
  \right),
  \qquad
  \sum_{a\in A}
  m(a)\psi'\left(
    u(a,\theta)-\zeta_\theta
  \right)
  =
  1, \label{BDP opt1}\\
  \sum_{\theta\in\Theta}
  \pi(\theta)\psi\left(
u(a,\theta)-\zeta_\theta
  \right)
  =
  \max_{b\in A}
  \sum_{\theta\in\Theta}
  \pi(\theta)\psi\left(
u(b,\theta)-\zeta_\theta
  \right)\text{ for all }a\in\supp m.\label{BDP opt2}
\end{gather}
Here \(m\) is the \(f\)-mean of \(P\), and \(\zeta_\theta\) is the
statewise normalizing constant.

The \(\alpha\)-information corresponds to the Csisz\'ar-information with 
$
  \varphi(t)\equiv   \varphi_\alpha(t)$, 
where 
\[
\varphi_\alpha(t)
  =
  \frac{4}{1-\alpha^2}
  \left(
    1-t^{(1-\alpha)/2}
    +\frac{1-\alpha}{2}(t-1)
  \right)
  \]
for \(\alpha\neq\pm1\), with limiting cases
\(\varphi_{-1}(t)=t\log t-t+1\) and
\(\varphi_1(t)=-\log t+t-1\).  With this choice, 
the $\alpha$-divergence is the \(f\)-divergence, and 
the $\alpha$-information is the Csisz\'ar information, while co-finiteness fails when $\alpha>-1$: 
\[
  \lim_{t\to\infty}\frac{\varphi_\alpha(t)}{t}
  =
  +\infty
  \quad\text{if } \alpha\leq -1,
  \qquad
  \lim_{t\to\infty}\frac{\varphi_\alpha(t)}{t}
  =
  \frac{2}{1+\alpha}
  \quad\text{if } \alpha>-1 .
\]

Thus, when \(\alpha\le -1\), the co-finiteness condition holds and BDP's 
optimality condition applies directly.  In this case,
\eqref{BDP opt1} and \eqref{BDP opt2} correspond to Parts 1 and 2 of
Proposition~\ref{main result prop}, respectively.  
Proposition~\ref{Amari proposition} uses the $\alpha$-integration to
identify the \(f\)-mean \(m\) explicitly.  This explicit expression for \(m\)
makes the optimality condition tractable in terms of the \(q\)-exponential formula.

When $\alpha>-1$, the co-finiteness condition fails, so BDP's
optimality condition does not directly apply. Nevertheless, an analogous
optimality condition holds, as established in Proposition~\ref{main result prop}.
Moreover, Proposition~\ref{main result prop} and Corollary~\ref{supp coro}
reveal that conditional choice supports exhibit different properties when
co-finiteness fails, as is also illustrated by the numerical example. Thus, our
analysis not only specializes BDP's optimality condition, but also demonstrates the
support behavior of optimal rules when co-finiteness fails.

\subsection{Posterior-separable power-form cost}\label{UPS}

We show that the optimal rule under the posterior-separable power-form cost \eqref{power cost} can also be represented in terms of the $q$-exponential, but that its behavioral implication differs from the optimal rule under the $\alpha$-information cost. 

A posterior-separable cost belongs to the general class of $f$-information costs, as shown by BDP. 
In particular, 
BDP's optimality condition implies the following representation of an optimal rule under \eqref{UPScost}: 
for each $a\in A$ with $P_\pi(a)= \sum_\theta P_\theta(a)\pi(\theta)>0$ and each $\theta\in \supp\pi$,
\[
P_\theta(a)
=
\frac{P_\pi(a)\nabla_\theta H^\star(x)}{\pi(\theta)},\quad 
x=(\kappa^{-1}(u(a,\theta')-\mu(\theta')/\pi(\theta')))_{\theta'\in \Theta}\in \R^\Theta,
\]
 where $\mu=(\mu(\theta))_{\theta\in \Theta}\in \R^\Theta$ is the multiplier in BDP's formulation, $H^\star$ is the convex conjugate of $H$, and $\nabla_\theta H^\star$ is the gradient of $H^\star$.  
When $H=H_s$, this is written as 
\begin{align}
P_\theta(a)
=
\frac{P_\pi(a)}{\pi(\theta)}
\exp_q
\left(
\frac{
u(a,\theta)-\eta_\theta-\zeta_a
}{\kappa}
\right),\quad \eta_\theta=\mu(\theta)/\pi(\theta), \label{UPS optimal}
\end{align}
where $q=2-s$ and $\zeta_a$ is the (state-independent) normalization constant for the posterior given by 
\[
\gamma_a^P(\theta)= 
\exp_q
\left(
\frac{
u(a,\theta)-\eta_\theta-\zeta_a
}{\kappa}
\right).
\]

The optimal rule \eqref{UPS optimal} and the optimal rule in Proposition
\ref{main result prop} are represented in terms of the same
$q$-exponential function when $q=2-s=q_\alpha$. 
However, there are three differences. First, \eqref{UPS optimal} contains
the action-specific constant $\zeta_a$, which shifts the argument of the
$q$-exponential function uniformly across states for each action. Such an
action-specific constant is absent in the optimal rule under the
$\alpha$-information cost, which takes the modified
$q$-logit form.  
Second, the two rules differ in their multiplicative factors outside the
$q$-exponential function. In \eqref{UPS optimal}, this factor is
$P_\pi(a)/\pi(\theta)$; in the optimal rule under the
$\alpha$-information cost, the corresponding factor is the
$\alpha$-integration $m_\alpha^P(a)$. 
Third, when $s<0$ and $q>2$, the conditional choice support of
\eqref{UPS optimal} coincides with the consideration set in each state,
which is not necessarily the case in the $\alpha$-information model.
The next lemma formally states this support property under the
posterior-separable power-form cost.

\begin{lemma}
Let $P$ be an optimal rule under the posterior-separable power-form cost \eqref{power cost}. 
Suppose $s<0$. If $P_\pi(a)>0$, then
$
P_\theta(a)>0$ for every $\theta\in\supp\pi$.
\end{lemma}

\begin{proof}
The cost of $\gamma_a^P$ must be finite, which requires that, for all $a\in \supp P_\pi$,  
$\gamma_a^P(\theta)>0$ for all $\theta\in\supp \pi$. 
Thus, 
$
P_\theta(a)
=
P_\pi(a)\gamma_a^P(\theta)/{\pi(\theta)}>0$ when $P_\pi(a)>0$ and $\pi(\theta)>0$. 
\end{proof}

This lemma implies that the posterior-separable power-form model cannot reproduce the 
optimal choice rules of the $\alpha$-divergence model for which the consideration set and
the conditional choice supports differ when $q_\alpha=2-s>2$.

\begin{appendix}

\section{Proofs}\label{secproofs}

\subsection{Proof of Proposition \ref{q-logit proposition}}

The objective function is 
\begin{align*}
\Phi_{\theta}(p,m)
&\equiv 
  \sum_{a\in A}p(a)u(a,\theta)
  -
  \kappa D_\alpha[p:m]\\
&  =
  \sum_{a\in A}p(a)u(a,\theta)
  -
  \frac{4\kappa}{1-\alpha^2}
  \left(
  1-
  \sum_{a\in A}
  p(a)^{(1-\alpha)/2}m(a)^{(1+\alpha)/2}
  \right),
\end{align*}
which is concave in $p$.  Hence the KKT
conditions are necessary and sufficient. 
The Lagrangian is
\[
  \mathcal L(p,\nu_\theta,\xi)
  =
\Phi_\theta(p,m)
  -
  \nu_\theta\left(\sum_{a\in A}p(a)-1\right)
  +
  \sum_{a\in A}\xi_a p(a).
\]
At differentiable points, the KKT conditions are
\[
  \sum_{a\in A}p(a)=1,
  \qquad
  p(a)\ge0,
\]
\[
  \xi_a\ge0,
  \qquad
  \xi_a p(a)=0,\qquad 
  \frac{\partial \Phi_\theta(p,m)}{\partial p(a)}
  -
  \nu_\theta+\xi_a=0\text{ for all }a\in A.
\]
At boundary points, this condition is replaced by the corresponding normal-cone condition.

For $m(a)>0$ and $p(a)>0$, direct differentiation gives, for
$\alpha\neq -1$,
\[
  \frac{\partial \Phi_\theta(p,m)}{\partial p(a)}
  =
  u(a,\theta)
  +
  \frac{2\kappa}{1+\alpha}
  \left(
  \frac{m(a)}{p(a)}
  \right)^{(1+\alpha)/2},
\]
where the case $\alpha=1$ is interpreted as the limiting reverse-KL case.
Thus, on every active coordinate with $m(a)>0$ and $p(a)>0$,
complementary slackness gives $\xi_a=0$, and hence
\begin{equation}
  u(a,\theta)
  +
  \frac{2\kappa}{1+\alpha}
  \left(
  \frac{m(a)}{p(a)}
  \right)^{(1+\alpha)/2}
  =
  \nu_\theta.	
\notag 
\end{equation}
For $\alpha\neq -1$, put
\[
  \lambda_\theta
  =
  \nu_\theta-\frac{2\kappa}{1+\alpha}.
\]
Then, the active-coordinate KKT condition is equivalent to
\begin{equation}
  p(a)
  =  m(a)
  \left[
  \frac{(1+\alpha)(\nu_\theta-u(a,\theta))}{2\kappa}
  \right]^{-2/(1+\alpha)}
=
  m(a)
  \exp_{q_\alpha}\left(
  \frac{u(a,\theta)-\lambda_\theta}{\kappa}
  \right).
\label{activeKKT}
\end{equation}

For $\alpha=-1$, the first-order condition is
\[
  u(a,\theta)
  -
  \kappa\left(\log\frac{p(a)}{m(a)}+1\right)
  =
  \nu_\theta,
\]
and hence
\[
  p(a)
  =
  m(a)
  \exp\left(
  \frac{u(a,\theta)-\nu_\theta}{\kappa}-1
  \right).
\]
This is the same $q_\alpha$-exponential formula with $q_\alpha=1$, after
absorbing the constant term into the multiplier.

Suppose that 
$\alpha< -1$. Then $(1+\alpha)/2<0$, so
finite cost requires that $
  p(a)=0$ whenever $m(a)=0$. 
Thus \eqref{activeKKT} applies on $A$, with zero probability assigned to
actions outside $S_m$, and requires that $u(a,\theta)>\nu_\theta$ for all $a\in S_m$ with $p(a)>0$. 
On the other hand, for $a\in S_m$ with $p(a)=0$, the KKT conditions imply 
\[
\xi_a\geq 0\quad\text{ and } \quad \frac{\partial \Phi_\theta(p,m)}{\partial p(a)}
  -
  \nu_\theta+\xi_a=
    u(a,\theta)    -
  \nu_\theta+\xi_a=  0
\]
and thus
$
  u(a,\theta)\leq \nu_\theta$.  
  Therefore, $a\in S_m$ satisfies $p(a)>0$ if and only if $
  u(a,\theta)> \nu_\theta$.
Therefore, \eqref{q ex family} holds, where $\lambda_\theta$ is determined by \eqref{normalization1}, i.e., $\lambda_\theta=\bar \lambda_{\theta m}$ and $\nu_\theta =\bar \nu_{\theta m}\equiv \bar \lambda_{\theta m}+2\kappa/(1+\alpha)$, and the solution is unique in this case.

Suppose that $\alpha>-1$ and $S_m=A$, where $(1+\alpha)/2>0$.  
For $a\in S_m$, ${\partial \Phi_\theta(p,m)}/{\partial p(a)}$ tends to $+\infty$ as
$p(a)\downarrow0$.  Hence $p(a)=0$ cannot be optimal, and every finite optimum satisfies $p(a)>0$ for all
$a\in A$.  Hence  \eqref{activeKKT} applies on $A$, and $\lambda_\theta$ is determined by \eqref{normalization1}, and the solution is unique also in this case.

Suppose that $\alpha>-1$ and $S_m\subsetneq A$.  
For each $\theta\in\Theta$, define
\begin{align}
  T_\theta(\nu)
  \equiv 
  \sum_{a\in S_m}
  m(a)
  \left[
  \frac{(1+\alpha)(\nu-u(a,\theta))}{2\kappa}
  \right]^{-2/(1+\alpha)}
=    \sum_{a\in S_m}
m(a)\exp_{q_\alpha}\left(
  \frac{u(a,\theta)-\lambda}{\kappa}
  \right),\notag 
\end{align}
where $\lambda=\nu-2\kappa/(1+\alpha)$. 
Note that $T_\theta(\nu)$ is strictly decreasing and $T_\theta(\bar\nu_{\theta m})=1$.

Since
$\alpha>-1$, every action in $S_m$ must receive positive probability at any
finite optimum by the above argument. Hence  \eqref{activeKKT} applies on $S_m$. 
In addition, the probability assigned to $S_m$ is equal to $T_\theta(\nu_\theta)$. 
Thus, we must have $T_\theta(\nu_\theta)\leq 1$, i.e.,  $\nu_\theta \geq \bar \nu_{\theta m}$. 
Note that $p(a)>0$ for some $a\not\in S_m$ if and only if  $T_\theta(\nu_\theta)< 1$, i.e.,  $\nu_\theta > \bar \nu_{\theta m}$ and $\lambda_\theta > \bar \lambda_{\theta m}$. 

Now consider actions outside $S_m$.  If $a\notin S_m$, then $m(a)=0$, so the divergence term including $a$ has no marginal contribution since $\alpha>-1$.  Hence the KKT condition becomes
$
  u(a,\theta)-\nu_\theta+\xi_a=0$, and complementary slackness implies
\[
  p(a)>0
\ \ \Rightarrow\ \ 
  u(a,\theta)=\nu_\theta \quad \text{ and }\quad 
  p(a)=0 \ \ 
\Rightarrow \ \ 
  u(a,\theta)\le\nu_\theta.
\]
Therefore, if $p(a)>0$ for some $a\not\in S_m$, we must have $u(a,\theta)=\max_{b\not\in S_m}u(b,\theta)=\nu_\theta > \bar \nu_{\theta m}$ and $\lambda_\theta=\bar\lambda_{\theta m}^0>\bar \lambda_{\theta m}$. 
Conversely, if this is true, such $p$ is optimal. 
This also implies that 
if $\bar\lambda_{\theta m}^0\leq \bar \lambda_{\theta m}$, then $p(a)=0$ for all $a\not\in S_m$, and thus $\nu_\theta=\bar\nu_{\theta m}$ and $\lambda_\theta=\bar\lambda_{\theta m}$ since $T_\theta(\nu_\theta)=1$. 
This proves the stated characterization.  The solution is unique if the set of
zero-$m$ maximizers is a singleton, and otherwise it is not unique.

\subsection{Proof of Proposition \ref{main result prop}}

Fix $P\in\Delta(A)^\Theta$. By Proposition~\ref{Amari proposition},
\[
  V(P)
  \equiv
  \max_{m\in\Delta(A)}\Phi(P,m)
  =
  \Phi(P,m_\alpha^P).
\]
Define
\[
  W(m)
\equiv 
  \max_{Q\in\Delta(A)^\Theta}\Phi(Q,m).
\]
Then
\begin{equation}
  \Phi(P,m_\alpha^P)
  \le
  W(m_\alpha^P)
  \le
  \max_{m\in\Delta(A)}W(m).
\label{auxineq}
\end{equation}
Moreover,
\[
\begin{aligned}
\max_{P\in\Delta(A)^\Theta}V(P)
&=
\max_{P\in\Delta(A)^\Theta}
\max_{m\in\Delta(A)}
\Phi(P,m)                                      \\
&=
\max_{m\in\Delta(A)}
\max_{P\in\Delta(A)^\Theta}
\Phi(P,m)                                      \\
&=
\max_{m\in\Delta(A)}W(m).
\end{aligned}
\]
Hence $P$ is optimal if and only if both inequalities in
\eqref{auxineq} hold with equality. Equivalently,
\begin{gather}
P\in
\operatorname*{argmax}_{Q\in\Delta(A)^\Theta}
\Phi(Q,m_\alpha^P),
\label{C1}\\
m_\alpha^P\in
\operatorname*{argmax}_{m\in\Delta(A)}W(m).
\label{C2}
\end{gather}
The first condition in the proposition is equivalent to \eqref{C1}. It
remains to show that the second condition is equivalent to \eqref{C2}.

Let 
\begin{equation}
  W_\theta(m)
  \equiv 
  \max_{p\in\Delta(A)}
  \Phi_\theta(p,m).
\notag 
\end{equation}
Then
$  W(m)
  =
  \sum_{\theta\in\Theta}\pi(\theta)W_\theta(m)$.
Note that $W$ is concave since, for each $\theta$,
$\Phi_\theta(p,m)$ is jointly concave in $(p,m)$ (recall that 
$D_\alpha[p:m]$ is jointly convex).


To solve \eqref{C2}, we evaluate the following one-sided directional derivatives: 
for $m,\tilde m\in\Delta(A)$, 
\[
  D^+W_\theta(m;\tilde m-m)
  \equiv 
\partial_t^+ 
    W_\theta((1-t)m+t\tilde m)\big|_{t=0},
\]
\[
  D^+W(m;\tilde m-m)
  \equiv 
\partial_t^+ 
    W((1-t)m+t\tilde m)\big|_{t=0}
=\sum_\theta\pi(\theta)  D^+W_\theta(m;\tilde m-m),
\]
where $\partial_t^+$ denotes the right derivative operator with respect to $t$.

\begin{lemma}\label{lemma2}
Let $\lambda_\theta(m)$ denote the
constant $\lambda_\theta$ determined in Proposition~\ref{q-logit proposition}. 
For every $m,\tilde m\in\Delta(A)$,
\begin{equation}
  D^+W_\theta(m;\tilde m-m)
  =
  d_\alpha
  \sum_{a\in\supp m\cup \supp\tilde m}
  (\tilde m(a)- m(a))h_\alpha\left[\exp_{q_\alpha}\left(
  \frac{u(a,\theta)-\lambda_\theta(m)}{\kappa}
  \right)\right],
\label{eq:statewise-directional-derivative}
\end{equation}
where 
\begin{equation}
  d_\alpha
  \equiv 
  \begin{cases}
  \displaystyle {2\kappa}/({1-\alpha}) & \text{ if } \alpha\ne1,\\
  \kappa & \text{ if } \alpha=1.
  \end{cases}
\notag 
\end{equation}
\end{lemma}

\begin{proof}
We write
$
  d\equiv \tilde m-m$ and 
  $m_t \equiv m+td=(1-t)m+t\tilde m$ for each $t\in[0,1]$.
Then, for each $t\in (0,1)$,
\[
\supp m_t=  S \equiv \supp m\cup\supp\tilde m,
\]
and $d(a)=\tilde m(a)-m(a)=0$ for all $a\notin S$. 
We define
$
  v(t)\equiv W_\theta(m_t)$ 
  and evaluate $\partial_t^+ v(t)$.

We use the following Danskin's theorem. Let \(K\) be compact, let
\(I\subset\mathbb R\) be an interval, and let \(f:K\times I\to\mathbb R\)
be continuous. Suppose that \(f(x,\cdot)\) is right-differentiable for each
\(x\in K\), and that \(\partial_t^+f(x,t)\) is continuous in \((x,t)\).
Then, for \(w(t)=\max_{x\in K}f(x,t)\),
\[
  \partial_t^+w(t)
  =
  \max_{x\in\argmax_{x'\in K}f(x',t)}
  \partial_t^+f(x,t).
\]

We apply Danskin's theorem to
$
v(t)=  W_\theta(m_t)
  =
  \max_{p\in\Delta(A)}
  \Phi_\theta(p,m_t)$ for $t\in(0,1)$. Note that only the coordinates in $S$
vary along $m_t=m+td$. 
Danskin's theorem gives
\[
\partial_t^+ v(t)
=  \partial_t^+ W_\theta(m_t)  =
  \max_{p\in  \arg\max_{p'} 
  \Phi_\theta(p',m_t)}
\partial_t^+ \Phi_\theta(p,m_t).
  \]
Since 
\[
\partial_t^+ \Phi_\theta(p,m_t)
  =
d_\alpha
  \sum_{a\in S}
  d(a)
  h_\alpha\left[\frac{p(a)}{m_t(a)}\right],
\]
Proposition \ref{q-logit proposition} yields 
\begin{align*}
\partial_t^+ v(t)
&=  d_\alpha
  \sum_{a\in S}
  d(a)
    h_\alpha\left[  \exp_{q_\alpha}\left(
  \frac{u(a,\theta)-\lambda_\theta(m_t)}{\kappa}
  \right)\right].
\label{eq:lemma2-derivative-along-segment}
\end{align*}

We show $\lim_{t\to 0}\lambda_\theta(m_t)=\lambda_\theta(m)$. 
First, suppose that either $\alpha\le -1$, or
$\alpha>-1$ and $S_m=A$.  In these cases, $\lambda_\theta(m_t)=\bar\lambda_{\theta m_t}$, which is continuous in $t$. Thus,  $\lim_{t\to 0}\lambda_\theta(m_t)=\lambda_\theta(m)$.

Next, suppose that
$
  \alpha>-1$ and $S_m=\supp m\subsetneq A$. 
Then, $\lambda_\theta(m_t)=\max\{\bar  \lambda_{\theta m_t},\bar  \lambda_{\theta m_t}^0\}$ by Proposition \ref{q-logit proposition}. 
Note that 
\[
\bar \lambda_{\theta m_t}^0
=\begin{cases}
\bar \lambda_{\theta m}^0=\displaystyle	\max_{a\not\in S_m}u(a,\theta)-2\kappa/(1+\alpha)&\text{ if }t=0,\\
\displaystyle	\max_{a\not\in S}u(a,\theta)-2\kappa/(1+\alpha)&\text{ if }t\in (0,1).\\
\end{cases}
\]
Since $S\supset S_m$, $\bar \lambda_{\theta m_t}^0\leq \bar \lambda_{\theta m}^0$. Define
\[
  \nu_\theta(m_t)
  \equiv 
  \lambda_\theta(m_t)+\frac{2\kappa}{1+\alpha}.
\]
We consider two cases.
\begin{enumerate}[\text{Case}  1:]
	\item Suppose that $\bar \lambda_{\theta m_t}^0=\bar \lambda_{\theta m}^0$ for $t\in (0,1)$. Then, $\lambda_\theta(m_t)$ is continuous in $t$ since $\bar\lambda_{\theta m_t}$ is continuous in $t$. 
	\item Suppose that $\bar \lambda_{\theta m_t}^0< \bar \lambda_{\theta m}^0$ for $t\in (0,1)$. Since 
	$\max_{a\not\in S_m}u(a,\theta)>	\max_{a\not\in S}u(a,\theta)$, there exists $a^*\in S\setminus S_m$ with $u(a^*,\theta)=\max_{a\not\in S_m}u(a,\theta)$. This requires that $\bar \lambda_{\theta m_t}>\bar \lambda_{\theta m}^0>\bar \lambda_{\theta m_t}^0$ by the definition of $\bar \lambda_{\theta m_t}$, and thus $\lambda_\theta(m_t)=\bar \lambda_{\theta m_t}$. 
	\begin{itemize}
	\item If $\bar \lambda_{\theta m}\geq \bar \lambda_{\theta m}^0$, then $\lambda_\theta(m_t)=\bar \lambda_{\theta m_t}\to \bar \lambda_{\theta m}=\lambda_\theta(m)$ as $t\to 0$.
	\item If $\bar \lambda_{\theta m}< \bar \lambda_{\theta m}^0$, then $\lambda_\theta(m)=\bar \lambda_{\theta m}^0$. 
    Let \(l\) be an arbitrary cluster point of \(\lambda_\theta(m_t)\) as \(t\to 0\).
If \(l>\bar\lambda_{\theta m}^0\), then along a convergent subsequence,
\[
  1=\sum_{a\in S}
  m_t(a)
  \exp_{q_\alpha}\left(
  \frac{u(a,\theta)-\lambda_\theta(m_t)}{\kappa}
  \right)
  \to \sum_{a\in S_m}
  m(a)
  \exp_{q_\alpha}\left(
  \frac{u(a,\theta)-l}{\kappa}
  \right)
\]
and 
\[
\sum_{a\in S_m}
  m(a)
  \exp_{q_\alpha}\left(
  \frac{u(a,\theta)-l}{\kappa}
  \right)<\sum_{a\in S_m}
  m(a)
  \exp_{q_\alpha}\left(
  \frac{u(a,\theta)-\bar\lambda_{\theta m}^0}{\kappa}
  \right)<1,
\]
which is a contradiction. Thus, we must have $\lim_{t\to 0} \lambda_\theta(m_t)=\bar \lambda_{\theta m}^0$. 
	\end{itemize}
\end{enumerate}
In all cases,  $\lim_{t\to 0}\lambda_\theta(m_t)=\lambda_\theta(m)$, and thus 
\begin{align*}
\lim_{t\to 0}\partial_t^+ v(t)
&= d_\alpha
  \sum_{a\in S}
  d(a)
      h_\alpha\left[\exp_{q_\alpha}\left(
  \frac{u(a,\theta)-\lambda_\theta(m)}{\kappa}
  \right)\right]\\  &=
  d_\alpha
  \sum_{a\in S}
  (\tilde m(a)-m(a))      h_\alpha\left[\exp_{q_\alpha}\left(
  \frac{u(a,\theta)-\lambda_\theta(m)}{\kappa}
  \right)\right].
\end{align*}

To complete the proof, it suffices to show that
$\partial_t^+v(0)=c\equiv \lim_{t\downarrow0}\partial_t^+v(t)$.
Let $\phi(t)\equiv (v(t)-v(0))/t$.  It suffices to show that
$\lim_{t\downarrow0}\phi(t)=c$.
Fix $\varepsilon>0$.  Then there exists $\delta>0$ such that
$c-\varepsilon<\partial_t^+v(t)<c+\varepsilon$ for all
$t\in(0,\delta)$.  For such $t$, concavity gives
$\partial_t^+v(t)\le \phi(t)$, and hence
$\phi(t)>c-\varepsilon$.
On the other hand, if $0<\eta<t<\delta$, concavity gives
\[
  \frac{v(t)-v(\eta)}{t-\eta}
  \le \partial_t^+v(\eta)<c+\varepsilon .
\]
Letting $\eta\downarrow0$ yields $\phi(t)\le c+\varepsilon$.  Therefore, for all
sufficiently small $t>0$,
$c-\varepsilon<\phi(t)\le c+\varepsilon$.  Since $\varepsilon>0$ is
arbitrary, $\lim_{t\downarrow0}\phi(t)=c$.  
\end{proof}

Summing \eqref{eq:statewise-directional-derivative} over states gives
\begin{equation}
  D^+W(m;\tilde m-m)
  =
  d_\alpha
  \sum_{a\in\supp m\cup \supp\tilde m}
  (\tilde m(a)- m(a))\sum_{\theta}\pi(\theta)h_\alpha\left[\exp_{q_\alpha}\left(
  \frac{u(a,\theta)-\lambda_\theta(m)}{\kappa}
  \right)\right].
\notag
\end{equation}
By concavity, \eqref{C2} is equivalent to
\begin{equation}
  D^+W(m_\alpha^P;\tilde m-m_\alpha^P)\le0
  \qquad
  \text{for every }\tilde m\in\Delta(A).
\label{directionalC2}
\end{equation}

Suppose first that $\alpha\le1$.  Then $d_\alpha>0$.  Hence \eqref{directionalC2} is equivalent to
\begin{align*}
\sum_{a\in\supp\tilde m}& 
  \tilde m(a)\sum_\theta\pi(\theta)h_\alpha\left[\exp_{q_\alpha}\left(
  \frac{u(a,\theta)-\lambda_\theta(m_\alpha^P)}{\kappa}
  \right)\right]\\
 & \le
  \sum_{a\in\supp m_\alpha^P}
  m_\alpha^P(a)\sum_\theta\pi(\theta)h_\alpha\left[\exp_{q_\alpha}\left(
  \frac{u(a,\theta)-\lambda_\theta(m_\alpha^P)}{\kappa}
  \right)\right]
	\end{align*}
for all $\tilde m$. 
This holds if and only if
\begin{align*}
\sum_{a\in\supp m_\alpha^P}
  m_\alpha^P(a)&\sum_\theta\pi(\theta)h_\alpha\left[\exp_{q_\alpha}\left(
  \frac{u(a,\theta)-\lambda_\theta(m_\alpha^P)}{\kappa}
  \right)\right]\\
&  =
  \max_{b\in A}\sum_\theta\pi(\theta)h_\alpha\left[\exp_{q_\alpha}\left(
  \frac{u(b,\theta)-\lambda_\theta(m_\alpha^P)}{\kappa}
  \right)\right].
\end{align*}
This implies that every $a\in \supp m_\alpha^P$ attains the maximum in the left-hand side.

Suppose next that $\alpha>1$.  Then $d_\alpha<0$.  
Then, similarly to the above, 
\eqref{directionalC2} holds if and only if
\begin{align*}
\sum_{a\in\supp m_\alpha^P}
  m_\alpha^P(a)&\sum_\theta\pi(\theta)h_\alpha\left[\exp_{q_\alpha}\left(
  \frac{u(a,\theta)-\lambda_\theta(m_\alpha^P)}{\kappa}
  \right)\right]\\
&  =
  \min_{b\in A}\sum_\theta\pi(\theta)h_\alpha\left[\exp_{q_\alpha}\left(
  \frac{u(b,\theta)-\lambda_\theta(m_\alpha^P)}{\kappa}
  \right)\right],
\end{align*}
and every $a\in \supp m_\alpha^P$ attains the minimum in the left-hand side. 

Thus \eqref{C2} is equivalent to the second condition.  Together with the
equivalence between \eqref{C1} and the first condition, this proves the
proposition.

\subsection{Proof of Corollary \ref{supp coro}}

If $\alpha<1$, Proposition~\ref{Amari proposition} implies
$
  S_m=\bigcup_{\theta\in\Theta}S_\theta$.
When $\alpha<-1$, Proposition~\ref{q-logit proposition} implies that no action outside
$S_m$ can receive positive probability, so
$
  S_\theta\subseteq S_m$ for every $\theta\in\Theta$.

Suppose that $-1\leq\alpha<1$.  Proposition~\ref{q-logit proposition} implies that every
action in $S_m$ receives positive probability in every state.  Thus,
$
  S_m\subseteq S_\theta$ for every $\theta\in\Theta$. 
Since Proposition~\ref{Amari proposition} also implies
$
  S_\theta\subseteq
  \bigcup_{\tau\in\Theta}S_\tau=S_m$, 
we obtain
$
  S_\theta=S_m$ for every $\theta\in\Theta$.

Suppose that $\alpha\geq1$.  Since an optimal stochastic
choice rule has finite $\alpha$-information, Proposition~\ref{Amari proposition} implies
$
  S_m=\bigcap_{\theta\in\Theta}S_\theta$.
Therefore,
$  S_m\subseteq S_\theta$ for every $\theta\in\Theta$.

The remaining statements directly follow from Proposition~\ref{q-logit proposition} and the definition of the $q$-exponential function.

\subsection{Proof of Proposition \ref{prop:guess-state}}

Write 
$
q\equiv q_\alpha$. 
We show that there exists a symmetric choice rule that satisfies the optimality condition and characterize it. 
Let $P$ be a symmetric choice rule. 
For each action $a$, the vector $P(a)=(P_\theta(a))_{\theta\in\Theta}$ is a permutation of
$
(\rho,(1-\rho)/(n-1),\ldots,(1-\rho)/(n-1))$. 
Hence $M_\alpha(P(a))$ is the same for every $a$, and the $\alpha$-integration is uniform:
$
m_\alpha^P(a)={1}/{n}$ for every $a\in A$. 
Therefore, by Proposition~\ref{q-logit proposition}, 
 $P$ is statewise-optimal if and only if 
\begin{equation}
\rho
=
\frac{1}{n}
\exp_q\left(
\frac{w-\lambda}{\kappa}
\right),
\qquad
\frac{1-\rho}{n-1}
=
\frac{1}{n}
\exp_q\left(
\frac{-\lambda}{\kappa}
\right),
\label{rhotau}
\end{equation}
where $\lambda$ is a solution to 
\begin{equation}
G_w(\lambda)\equiv\exp_q\left(
\frac{w-\lambda}{\kappa}
\right)
+
(n-1)
\exp_q\left(
\frac{-\lambda}{\kappa}
\right)=n.
\notag\label{normal guess}	
\end{equation}
Moreover, by symmetry, the second condition in Proposition~\ref{main result prop} is automatically satisfied; hence, for symmetric rules, optimality is equivalent to the statewise optimality condition in Proposition~\ref{q-logit proposition}.
Since $\lambda$ exists and is unique, the symmetric optimal choice rule is unique.

First suppose that $\alpha<-1$, so $q<1$.
Then,
\[
\exp_q\left(
\frac{-\lambda}{\kappa}
\right)>0
\quad\Longleftrightarrow\quad
1-(1-q)\frac{\lambda}{\kappa}>0
\quad\Longleftrightarrow\quad
\lambda<\lambda_0\equiv\frac{\kappa}{1-q}.
\]
Since $G_w(\lambda)$ is strictly decreasing in $\lambda$, 
$\lambda\geq \lambda_0$ if and only if $G_w(\lambda_0)\geq n$, which is rewritten as 
\[
\left(
(1-q)\frac{w}{\kappa}
\right)^{1/(1-q)}
\geq 
n\quad \Leftrightarrow\quad 
w\geq \frac{\kappa n^{1-q}}{1-q}
=
\bar w_\alpha.
\]
In this case,  $\rho=1$. 
If $w<\bar w_\alpha$, then $\rho<1$. 

When \(\rho<1\), 
 \eqref{eq:guess-state-response} and \eqref{eq:guess-state-response-q1} are derived from \eqref{rhotau}. To see this, note that \eqref{rhotau} is equivalent to
\[
\log_q(n\rho)=\frac{w-\lambda}{\kappa},
\qquad
\log_q\left(\frac{n(1-\rho)}{n-1}\right)=\frac{-\lambda}{\kappa}.
\]
Eliminating \(\lambda\) gives
\begin{align}
\frac{w}{\kappa}
&=
\log_q(n\rho)
-
\log_q\left(\frac{n(1-\rho)}{n-1}\right)=
\begin{cases}
	\dfrac{
(n\rho)^{1-q}
-
\left(\dfrac{n(1-\rho)}{n-1}\right)^{1-q}
}{1-q}& \text{ if }q\neq 1,\\
\log\left(\dfrac{(n-1)\rho}{1-\rho}\right) & \text{ if }q=1,
\end{cases}
\label{keyeq guess}
\end{align}
which implies \eqref{eq:guess-state-response} and \eqref{eq:guess-state-response-q1}, respectively.

\subsection{Proof of Corollary \ref{prop:guess-state-asymptotic}}

If \(\alpha>-1\) (i.e. \(q_\alpha>1\)), then
by \eqref{keyeq guess}, 
\[
\frac{w}{\kappa}
=
\frac{1}{q_\alpha-1}
\left\{
\left(\frac{n(1-\rho(w))}{n-1}\right)^{1-q_\alpha}
-
(n\rho(w))^{1-q_\alpha}
\right\}
\leq
\frac{1}{q_\alpha-1}
\left(\frac{n(1-\rho(w))}{n-1}\right)^{1-q_\alpha}.
\]
Since \(1-q_\alpha<0\), it follows that
\[
1-\rho(w)
\leq
\frac{n-1}{n}
(q_\alpha-1)^{-1/(q_\alpha-1)}
\left(\frac{w}{\kappa}\right)^{-1/(q_\alpha-1)}=O\left(\left(\frac{w}{\kappa}\right)^{-\frac{2}{1+\alpha}}\right).
\]

If \(\alpha=-1\) (i.e. \(q_\alpha=1\)), then 
\[
1-\rho(w)
=
\frac{n-1}{\exp(w/\kappa)+n-1}
\leq
(n-1)\exp\left(-\frac{w}{\kappa}\right)
=
O\left(\exp\left(-\frac{w}{\kappa}\right)\right).
\]

\subsection{Proof of Lemma \ref{lem:elasticity-marginal-cost-alpha}}

Put
$
r={\rho}/({1-\rho})>1$ and $
s={(1+\alpha)}/{2}>0$. 
Then
\[
\calE_{\alpha}(\rho)
=
s\frac{r^{s+1}+1}{r^s-1}.
\]
It suffices to show that this expression is strictly increasing in $s$.

Let
$
x=r^s$. 
Then $x>1$, and differentiation with respect to $s$ gives
\begin{equation}
\frac{d}{ds}
\left[
s\frac{r^{s+1}+1}{r^s-1}
\right]
=
\frac{
(rx+1)(x-1)-(r+1)x\log x
}{
(x-1)^2
}.
\label{els eq lemma1}
\end{equation}
The denominator is strictly positive. The numerator satisfies
\begin{align}
(rx+1)(x-1)-(r+1)x\log x
&=
rx(x-1-\log x)+x-1-x\log x\notag
\\
&>
x(x-1-\log x)+x-1-x\log x\notag
\\
&=
x^2-1-2x\log x>0.\notag\label{elas eq lemma}
\end{align}
The first strict inequality follows from $r>1$ and $x-1-\log x>0$ for $x>1$.
The last strict inequality holds since $x^2-1-2x\log x=0$ at $x=1$ and  
\[
\frac{d}{dx}\left(x^2-1-2x\log x\right)
=
2(x-1-\log x)>0.
\]
Therefore, \eqref{els eq lemma1} is strictly positive.




\section{Details of the estimation in Section \ref{Estimated magnitude sec}}
\label{app:individual-profile-likelihood}

This appendix records 
the numerical implementation of the profile-likelihood estimator.

\subsection*{Data}
\label{app:data}

The data are the individual choices from Experiment~1.2 of
\citet{deanNeligh2023}. 
Each trial is classified as correct or incorrect, and the reward for a correct
choice $w\in\{5,40,70,95\}$ is measured in probability points. 
For subject $i$ and reward $w$, let $n_{iw}$ be the number of trials
and $y_{iw}$ the number of correct choices.
There are
$52$ subjects, and 
 $
n_{iw}=50$ for each $i$ and $w$, so each subject contributes $200$
trials.
The aggregate empirical correct-choice frequencies are
\[
\begin{array}{c|cccc}
\text{reward} & 5 & 40 & 70 & 95\\
\hline
\text{relative frequency} & 0.671 & 0.711 & 0.720 & 0.753
\end{array}
\]
These frequencies are increasing in $w$ but only mildly so, which is
the weak responsiveness to incentives noted in the main text.

\subsection*{Response function and reparameterization}
\label{app:response}

Define 
\[
F_\alpha(\rho)
\equiv
\log_{q_\alpha}(2\rho)-\log_{q_\alpha}\qty(2(1-\rho)).
\]
Then $\rho(w|\alpha,\kappa_i)$ is the unique solution to 
$F_\alpha(\rho)=w/\kappa_i$ by \eqref{eq:guess-state-response}. 
Note that 
$\lim_{\kappa_i\to\infty}\rho(w|\alpha,\kappa_i)=1/2$
and $\lim_{\kappa_i\to 0}\rho(w|\alpha,\kappa_i)=1$ for every $w$. 
Since the range of $\kappa_i$ is not bounded, we adopt 
\[
\mu_i
\equiv
\rho(95|\alpha,\kappa_i)
\in
\left[1/2,1\right]
\]
as the nuisance parameter instead of $\kappa_i(={95}/{F_\alpha(\mu_i)})$. 
The parameter $\mu_i\in[1/2,1]$ produces the same one-dimensional family of response functions as $\kappa_i\in[0,\infty)$.
The fitted correct-choice probability is then obtained from
\begin{equation}\label{rho est eq}
\rho_{iw}
=F_\alpha^{-1}\qty(
w F_\alpha(\mu_i)/95).
\end{equation}

\subsection*{Profile likelihood}

For each subject $i$, the grouped binomial log likelihood, omitting constants that do not depend on parameters, is
\[
\ell_i(\alpha,\mu_i)
=
\sum_{w\in\{5,40,70,95\}}
\left[
y_{iw}\log \rho_{iw}
+
(n_{iw}-y_{iw})\log(1-\rho_{iw})
\right],
\]
where $\rho_{iw}$ is obtained from \eqref{rho est eq}.
For each fixed $\alpha$, the profile log likelihood is
\[
\ell_p(\alpha)
=
\sum_{i=1}^{52}
\max_{\mu_i\in[1/2,1]}
\ell_i(\alpha,\mu_i).
\]
Each subject's likelihood is maximized over the
compact interval $\mu_i\in[1/2,1]$ by a one-dimensional search. 
The profile maximum likelihood estimate
$\widehat\alpha=\arg\max_{\alpha\ge -1}\ell_p(\alpha)$ is then obtained
by a one-dimensional search over $\alpha$.

\subsection*{Estimates}
\label{app:estimates}

\begin{table}[t]
\centering
\caption{Selected values of the profile log likelihood}
\label{tab:profile-likelihood-app}
\begin{tabular}{c|cc}
\hline
$\alpha$ & $\ell_p(\alpha)$ & $2\{\ell_p(\widehat\alpha)-\ell_p(\alpha)\}$\\
\hline
$-1.0000$ & $-5273.5932$ & $349.1594$\\
$0.0000$ & $-5155.4056$ & $112.7842$\\
$1.0000$ & $-5117.5263$ & $37.0254$\\
$2.0000$ & $-5102.4182$ & $6.8093$\\
$2.2437$ & $-5100.9343$ & $3.8415$\\
$3.0000$ & $-5099.0633$ & $0.0994$\\
$3.1751$ & $-5099.0136$ & $0.0000$\\
$4.0000$ & $-5099.7455$ & $1.4639$\\
$4.7319$ & $-5100.9343$ & $3.8415$\\
$5.0000$ & $-5101.4400$ & $4.8529$\\
$6.0000$ & $-5103.7819$ & $9.5368$\\
$8.0000$ & $-5109.6427$ & $21.2583$\\
$10.0000$ & $-5115.0759$ & $32.1248$\\
\hline
\end{tabular}
\end{table}

The profile log-likelihood is maximized at
\[
\widehat\alpha=3.1751
\qquad(\text{$3.18$ in the main text}),
\qquad
\ell_p(\widehat\alpha)=-5099.01 .
\]
At the mutual-information benchmark $\alpha=-1$,
$\ell_p(-1)=-5273.6$, so the likelihood-ratio statistic against the
benchmark is $2\{\ell_p(\widehat\alpha)-\ell_p(-1)\}\approx 349$.
Table~\ref{tab:profile-likelihood-app} reports the profile log-likelihood at
selected values of $\alpha$.

At $\widehat\alpha=3.1751$, the estimated subject-level nuisance parameters are classified as in Table~\ref{tab:boundary-classification-app}. 
Forty subjects are interior, $0<\widehat\kappa_i<\infty$.
Three subjects are fitted at the boundary $\widehat\kappa_i=0$, i.e.\
$\widehat\rho_{iw}=1$ for all $w$. 
The remaining nine subjects are fitted at the opposite limit
$\widehat\kappa_i\to\infty$, i.e.\ $\widehat\rho_{iw}=1/2$ for all
$w$.\footnote{The three subjects at $\widehat\kappa_i=0$ have
identifiers $2117$, $2201$, $2208$; the nine chance-level subjects have
identifiers $2108$, $2115$, $2123$, $2124$, $2204$, $2207$, $2319$,
$2321$, $2323$.}
At either limit the implied probabilities are $1/2$ or $1$ and independent of $\alpha$.
Consequently they affect neither $\widehat\alpha$ nor the
likelihood-ratio statistic: re-estimating $\alpha$ on the forty
interior subjects alone yields the identical
$\widehat\alpha=3.1751$.
The estimate of $\alpha$ is therefore identified entirely from the
common shape of the response function across the interior subjects.

\begin{table}[t]
\centering
\caption{Boundary classification at $\widehat\alpha=3.1751$}
\label{tab:boundary-classification-app}
\begin{tabular}{c|c}
\hline
Estimated status & Number of subjects\\
\hline
$\widehat\kappa_i=\infty$ & $9$\\
$0<\widehat\kappa_i<\infty$ & $40$\\
$\widehat\kappa_i=0$ & $3$\\
\hline
\end{tabular}
\end{table}

\subsection*{Common-cost specification}

For comparison with the estimate translated from \citet{deanNeligh2023}, we also estimate the common-cost specification.
In this specification, both $\alpha$ and $\kappa$ are common across subjects.
Equivalently, there is a single value
$
\mu=\rho(95|\alpha,\kappa)$ 
rather than subject-specific values $\mu_i$.
The log likelihood is
\[
\ell^{\mathrm{cc}}(\alpha,\mu)
=
\sum_i
\sum_{w\in\{5,40,70,95\}}
\left[
y_{iw}\log \rho_w(\alpha,\mu)
+
(n_{iw}-y_{iw})\log\{1-\rho_w(\alpha,\mu)\}
\right],
\]
where
$
F_\alpha(\rho_w(\alpha,\mu))
=
{w}F_\alpha(\mu)/{95}$. 
Maximizing this likelihood gives
\[
\widehat\alpha^{\mathrm{cc}}
\simeq
23.95,
\qquad
\widehat\kappa^{\mathrm{cc}}
={95}/{F_{\widehat\alpha^{\mathrm{cc}}}(\widehat \mu^{\mathrm{cc}})}
\simeq
0.382.
\]

The common-cost log-likelihood is very flat in $\alpha$ above
$\alpha\approx 15$: $\ell^{\mathrm{cc}}$ changes
by less than one log-likelihood point over the range
$\alpha\in[20,28]$, so $\widehat\alpha^{\mathrm{cc}}$ is only weakly
identified at the top of the range.
As explained in the main text, the gap between
$\widehat\alpha^{\mathrm{cc}}\simeq 24$ and $\widehat\alpha\simeq 3.18$
reflects the inability of a single $\kappa$ to absorb subject-level
differences in the overall cost scale.

\end{appendix}


\end{document}